\documentclass{article}

\usepackage{amsfonts}
\usepackage{amsmath}
\usepackage{amsthm}
\usepackage{tikz}

\usepackage[margin=2.5cm]{geometry}

\newtheorem{theorem}{Theorem}[section]
\newtheorem{lemma}[theorem]{Lemma}

\newtheorem{corollary}[theorem]{Corollary}

\newtheorem{definition}[theorem]{Definition}

\newcommand{\partdiff}[2]{\frac{\partial {#1}}{\partial {#2}}}

\newcommand{\mixdiff}[3]{\frac{\partial^2 {#1}}{{\partial {#2}}{\partial {#3}}}}

\def\E{{\bf E}}
\def\b1{{\bf 1}}
\def\bc{{\bf c}}
\def\be{{\bf e}}
\def\bx{{\bf x}}
\def\by{{\bf y}}
\def\bz{{\bf z}}
\def\bq{{\bf q}}
\def\bu{{\bf u}}
\def\bv{{\bf v}}
\def\RR{{\boldmath R}}
\def\ZZ{{\boldmath Z}}
\def\cI{{\cal I}}
\def\cM{{\cal M}}
\def\cB{{\cal B}}
\def\cF{{\cal F}}
\def\cG{{\cal G}}
\def\DD{{D}}

\title{Symmetry and Approximability of Submodular Maximization Problems\thanks{
An extended abstract \cite{Vondrak09} appeared in IEEE FOCS 2009.
This work was done while the author was at Princeton University.}}

\author{Jan Vondr\'ak \thanks{IBM Almaden Research Center, San Jose, CA 95120 ({\tt jvondrak@us.ibm.com}).}}

\begin{document}

\maketitle

\begin{abstract}
A number of recent results on optimization problems involving submodular
functions have made use of the {\em multilinear relaxation} of the problem.
These results hold typically in the {\em value oracle} model, where the objective
function is accessible via a black box returning $f(S)$ for a given $S$.
We present a general approach to deriving inapproximability results in the
value oracle model, based on the notion of {\em symmetry gap}. Our main result
is that for any fixed instance that exhibits a certain symmetry gap in
its multilinear relaxation, there is a naturally related class of instances
for which a better approximation factor than the symmetry gap would
require exponentially many oracle queries.
This unifies several known hardness results for submodular maximization,
e.g. the optimality of $(1-1/e)$-approximation for
monotone submodular maximization under a cardinality constraint,
and the impossibility of $(\frac12+\epsilon)$-approximation
for unconstrained (nonmonotone) submodular maximization.

As a new application, we consider the problem of maximizing a
nonmonotone submodular function over the bases of a matroid.
A $(\frac16-o(1))$-approximation has been developed for this problem, assuming
that the matroid contains two disjoint bases. 
We show that the best approximation one can achieve is indeed
related to packings of bases in the matroid. Specifically, for any $k \geq 2$,
there is a class of matroids of fractional base packing number $\nu = \frac{k}{k-1}$,
such that any algorithm achieving a better than $(1-\frac{1}{\nu}) = \frac{1}{k}$-approximation
for this class would require exponentially many value queries.
In particular, there is no constant-factor approximation for maximizing
a nonmonotone submodular function over the bases of a general matroid.
On the positive side, we present
a $\frac12 (1-\frac{1}{\nu}-o(1))$-approximation algorithm assuming
fractional base packing number at least $\nu$ where $\nu \in (1,2]$.
We also present an improved $0.309$-approximation for maximization of
a nonmonotone submodular function subject to a matroid independence constraint,
improving the previously known factor of $\frac14-\epsilon$. 
For this problem, we obtain a hardness of $(\frac12 + \epsilon)$-approximation
for any fixed $\epsilon>0$.
\end{abstract}




\section{Introduction}

Submodular set functions are defined by the following condition for all pairs of sets $S,T$:
$$ f(S \cup T) + f(S \cap T) \leq f(S) + f(T),$$
or equivalently by the property that the {\em marginal value} of any element,
$f_S(j) = f(S+j)-f(S)$, satisfies $f_T(j) \leq f_S(j)$, whenever
$j \notin T \supset S$. In addition, a set function is called monotone
if $f(S) \leq f(T)$ whenever $S \subseteq T$.
Throughout this paper, we assume that $f(S)$ is nonnegative.

Submodular functions have been studied in the context of combinatorial
optimization since the 1970's, especially in connection with matroids
\cite{E70,E71,NWF78,NWF78II,NW78,W82a,W82b,L83,F97}.
Submodular functions appear mostly for the following two reasons:
(i) submodularity arises naturally in various combinatorial settings,
and many algorithmic applications use it either explicitly or implicitly;
(ii) submodularity has a natural interpretation as the property
of {\em diminishing returns}, which defines an important class of
utility/valuation functions.
Submodularity as an abstract concept is both general enough to be useful for applications
and it carries enough structure to allow strong positive results.
A fundamental algorithmic result is that any submodular function
can be {\em minimized} in strongly polynomial time \cite{FFI01,Lex00}.

In contrast to submodular minimization, submodular maximization problems are
typically hard to solve exactly. Consider the classical problem
of maximizing a monotone submodular function subject to a cardinality constraint,
$\max \{f(S): |S|\leq k\}$. It is known that this problem admits a $(1-1/e)$-approximation
\cite{NWF78} and this is optimal for two different reasons:
(i) Given only black-box access to $f(S)$, we cannot achieve a better approximation,
unless we ask exponentially many value queries \cite{NW78}.
This holds even if we allow unlimited computational power.
(ii) In certain special cases where $f(S)$ has a compact representation on the input,
it is NP-hard to achieve an approximation better than $1-1/e$ \cite{Feige98}.
The reason why the hardness threshold is the same in both cases
is apparently not well understood.

An optimal $(1-1/e)$-approximation for the problem $\max \{f(S): |S| \leq k\}$ where $f$ is
monotone sumodular is achieved
by a simple greedy algorithm \cite{NWF78}. This seems to be rather coincidental; for
other variants of submodular maximization, such as unconstrained (nonmonotone)
submodular maximization \cite{FMV07}, monotone submodular maximization subject
to a matroid constraint \cite{NWF78II,CCPV07,Vondrak08}, or submodular
maximization subject to linear constraints \cite{KTS09,LMNS09},
greedy algorithms achieve suboptimal results. A tool which has proven useful
in approaching these problems is the {\em multilinear relaxation}.

\

\noindent{\bf Multilinear relaxation.}
Let us consider a discrete optimization problem $\max \{f(S): S \in \cF\}$, where
$f:2^X \rightarrow \RR$ is the objective function and $\cF \subset 2^X$
is the collection of feasible solutions. In case $f$ is a linear function,
$f(S) = \sum_{j \in S} w_j$, it is natural to replace this problem by a linear
programming problem. For a general set function $f(S)$, however, a linear
relaxation is not readily available. Instead, the following relaxation
has been proposed \cite{CCPV07,Vondrak08,CCPV09}:
For $\bx \in [0,1]^X$, let $\hat{\bx}$ denote a random vector in $\{0,1\}^X$
where each coordinate of $x_i$ is rounded independently to $1$ with probability $x_i$
and $0$ otherwise.\footnote{We denote vectors consistently in boldface $(\bx)$ and their coordinates in italics $(x_i)$.
We also identify vectors in $\{0,1\}^n$ with subsets of $[n]$ in a natural way.
}
We define
$$ F(\bx) = \E[f(\hat{\bx})] = \sum_{S \subseteq X} f(S) \prod_{i \in S} x_i
 \prod_{j \notin S} (1-x_j).$$
This is the unique {\em multilinear polynomial} which coincides with $f$ on $\{0,1\}$-vectors.
We remark that although we cannot compute the exact value of $F(\bx)$ for a given
$\bx \in [0,1]^X$ (which would require querying $2^n$ possible values of $f(S)$),
we can compute $F(\bx)$ approximately, by random sampling. Sometimes this causes
technical issues, which we also deal with in this paper.

Instead of the discrete problem $\max \{f(S): S \in \cF\}$,
we consider a continuous optimization problem $\max \{F(\bx): \bx \in P(\cF)\}$,
where $P(\cF)$ is the convex hull of characteristic vectors corresponding to $\cF$,
$$ P(\cF) 
 = \left\{ \sum_{S \in \cF} \alpha_S \b1_S:
   \sum_{S \in \cF} \alpha_S = 1, \alpha_S \geq 0 \right\}.$$
The reason why the extension $F(\bx) = \E[f(\hat{\bx})]$ is useful
for submodular maximization problems is that $F(\bx)$ has convexity properties
that allow one to solve the continuous problem $\max \{F(\bx): \bx \in P\}$
(within a constant factor) in a number of interesting cases.
Moreover, fractional solutions can be often
rounded to discrete ones without losing {\em anything} in terms of the objective
value. Then, our ability to solve the multilinear relaxation
translates directly into an algorithm for the original problem.
In particular, this is true when the collection of feasible solutions
forms a matroid.

{\em Pipage rounding} was originally developed by Ageev and Sviridenko
for rounding solutions in the bipartite matching polytope \cite{AS04}.
The technique was adapted to matroid polytopes by Calinescu et al.
 \cite{CCPV07}, who proved that for any submodular function $f(S)$
and any $\bx$ in the matroid base polytope $B(\cM)$,
the fractional solution $\bx$ can be rounded to a base
$B \in \cB$ such that $f(B) \geq F(\bx)$.
This approach leads to an optimal $(1-1/e)$-approximation for
the Submodular Welfare Problem,
and more generally for monotone submodular maximization subject
to a matroid constraint \cite{CCPV07,Vondrak08}.
It is also known that
the factor of $1-1/e$ is optimal for the Submodular Welfare Problem
both in the NP framework \cite{KLMM05} and in the value oracle model
\cite{MSV08}. Under the assumption that the submodular function $f(S)$
has {\em curvature} $c$, there is a $\frac{1}{c}(1-e^{-c})$-approximation
and this is also optimal in the value oracle model \cite{VKyoto08}.
The framework of pipage rounding can be also extended to nonmonotone submodular functions;
this presents some additional issues which we discuss in this paper.

For the problem of
unconstrained (nonmonotone) submodular maximization, a $2/5$-approximation
was developed in \cite{FMV07}. This algorithm implicitly uses
the multilinear relaxation $\max \{F(\bx): \bx \in [0,1]^X\}$.
For {\em symmetric} submodular functions, it is shown in \cite{FMV07}
that a uniformly random solution $\bx = (1/2,\ldots,1/2)$ gives
$F(\bx) \geq \frac12 OPT$, and there is no better approximation algorithm
in the value oracle model. Recently, a $1/2$-approximation was found
for unconstrained maximization of a general nonnegative submodular function \cite{BFNS12}.
This algorithm can be formulated in the multilinear relaxation framework, but also
as a randomized combinatorial algorithm.

Using additional techniques, the multilinear relaxation can be also applied
to submodular maximization with knapsack constraints
($\sum_{j \in S} c_{ij} \leq 1$). For the problem of maximizing a monotone
submodular function subject to a constant number of knapsack constraints,
there is a $(1-1/e-\epsilon)$-approximation algorithm for any $\epsilon > 0$ \cite{KTS09}.
For maximizing a nonmonotone submodular function
subject to a constant number of knapsack constraints,
a $(1/5-\epsilon)$-approximation was designed in \cite{LMNS09}.

One should mention that not all the best known results for submodular maximization
have been achieved using the multilinear relaxation. The greedy algorithm
yields a $1/(k+1)$-approximation for monotone submodular maximization
subject to $k$ matroid constraints \cite{NWF78II}. Local search methods
have been used to improve this to a $1/(k+\epsilon)$-approximation,
and to obtain a $1/(k+1+1/(k-1)+\epsilon)$-approximation
for the same problem with a nonmonotone submodular function,
for any $\epsilon>0$ \cite{LMNS09,LSV09}.
For the problem of maximizing a nonmonotone submodular function
over the {\em bases} of a given matroid, local search yields
a $(1/6-\epsilon)$-approximation, assuming that the matroid contains two disjoint bases
\cite{LMNS09}.

\subsection{Our results}

Our main contribution (Theorem~\ref{thm:general-hardness}) 
is a general hardness construction that yields
inapproximability results in the value oracle model in an automated
way, based on what we call the {\em symmetry gap} for some fixed instance.
In this generic fashion, we are able to replicate a number of previously
known hardness results (such as the optimality of the factors
$1-1/e$ and $1/2$ mentioned above),
and we also produce new hardness results using this construction
(Theorem~\ref{thm:matroid-bases}).
Our construction helps explain the particular hardness thresholds obtained under
various constraints, by exhibiting a small instance where the threshold can be
seen as the gap between the optimal solution and the best symmetric solution.
The query complexity results in \cite{FMV07,MSV08,VKyoto08}
can be seen in hindsight as special cases of Theorem~\ref{thm:general-hardness},
but the construction in this paper is somewhat different and
technically more involved than the previous proofs for particular cases.

\paragraph{Concrete results}
Before we proceed to describe our general hardness result, we present its implications
for two more concrete problems.
We also provide closely matching approximation algorithms for these two problems,
based on the multilinear relaxation. In the following, we assume that the objective function is given by a value oracle
and the feasibility constraint is given by a membership oracle: a value oracle for $f$ returns the value $f(S)$ for a given set $S$,
and a membership oracle for $\cF$ answers whether $S \in \cF$ for a given set $S$.

First, we consider the problem of maximizing a nonnegative (possibly nonmonotone) submodular function
subject to a \emph{matroid base} constraint. (This generalizes for example the maximum bisection
problem in graphs.) We show that the approximability of this problem is related to base
packings in the matroid. We use the following definition.

\begin{definition}
\label{def:fractional-base}
For a matroid $\cM$ with a collection of bases $\cB$, the fractional base
packing number is the maximum possible value of $\sum_{B \in \cB} \alpha_B$
for $\alpha_B \geq 0$ such that $\sum_{B \in \cB: j \in B} \alpha_B \leq 1$ for
every element $j$.
\end{definition}

\noindent{\bf Example.}
Consider a uniform matroid with bases $\cB = \{ B \subset [n]: |B| = k \}$. ($[n]$ denotes the set of integers
$\{1,2,\cdots,n\}$.)
Here, we can take each base with a coefficient $\alpha_B = 1 / {n-1 \choose k-1}$, which satisfies
the condition $\sum_{B \in \cB: j \in B} \alpha_B \leq 1$ for every element $j$ since every element
is contained in ${n-1 \choose k-1}$ bases. We obtain that the fractional packing number is at least
${n \choose k} / {n-1 \choose k-1} = \frac{n}{k}$. It is also easy to check that the fractional packing number
cannot be larger than $\frac{n}{k}$.

\begin{theorem}
\label{thm:matroid-bases}
For any $\nu$ in the form $\nu = \frac{k}{k-1}, k \geq 2$,
and any fixed $\epsilon>0$, a $(1-\frac{1}{\nu}+\epsilon) = \frac{1}{k}$-approximation for the problem
$\max \{f(S): S \in \cB\}$, where $f(S)$ is a nonnegative submodular function,
and $\cB$ is a collection of bases in a matroid with fractional packing number at least $\nu$,
would require exponentially many value queries.

On the other hand, for any $\nu \in (1,2]$,
there is a randomized $\frac{1}{2}(1-\frac{1}{\nu}-o(1))$-approximation
for the same problem.
\end{theorem}

In case the matroid contains two disjoint bases ($\nu=2$), we obtain a
$(\frac14-o(1))$-approximation, improving the previously known factor of $\frac16-o(1)$ \cite{LMNS09}.
In the range of $\nu \in (1,2]$, our positive and negative results are within a factor of $2$.
For maximizing a submodular function over the bases of a general matroid, 
we obtain the following.

\begin{corollary}
\label{coro:matroid-bases}
For the problem $\max \{f(S): S \in \cB\}$,  where $f(S)$ is a nonnegative submodular function,
and $\cB$ is a collection of bases in a matroid, any constant-factor approximation requires
an exponential number of value queries.
\end{corollary}

We also consider the problem of maximizing a nonnegative submodular function subject to a matroid independence
constraint. 

\begin{theorem}
\label{thm:matroid-indep}
For any $\epsilon>0$, a $(\frac12 + \epsilon)$-approximation for the problem $\max \{f(S): S \in \cI \}$,
where $f(S)$ is a nonnegative submodular function, and $\cI$ is a collection of independent sets in a matroid,
would require exponentially many value queries.

On the other hand, there is a randomized $\frac14 (-1+\sqrt{5}-o(1)) \simeq 0.309$-approximation for the same problem.
\end{theorem}

Our algorithmic result improves a previously known $(\frac14 - o(1))$-approximation \cite{LMNS09}.
The hardness threshold follows from our general result, but also quite easily from \cite{FMV07}.

\medskip
\noindent{\bf Hardness from the symmetry gap.}
Now we describe our general hardness result.
Consider an instance $\max \{f(S): S \in \cF\}$ which exhibits a certain
degree of symmetry. This is formalized by the notion of a {\em symmetry group}
$\cG$. We consider permutations $\sigma \in {\bf S}(X)$ where ${\bf S}(X)$
is the symmetric group (of all permutations) on the ground set $X$.
We also use $\sigma$ for the naturally induced mapping of subsets of $X$:
$\sigma(S) = \{ \sigma(i): i \in S \}$.
We say that the instance is invariant
under $\cG \subset {\bf S}(X) $, if for any $\sigma \in \cG$
and any $S \subseteq X$,
$f(S) = f(\sigma(S))$ and $S \in \cF \Leftrightarrow \sigma(S) \in \cF$.
We emphasize that even though we apply $\sigma$ to sets, it must be derived
from a permutation on $X$.
For $\bx \in [0,1]^X$, we define the ``symmetrization of $\bx$'' as
$$\bar{\bx} = \E_{\sigma \in {\cal G}}[\sigma(\bx)],$$
where $\sigma \in \cG$ is uniformly random and $\sigma(\bx)$ denotes $\bx$ with coordinates permuted by $\sigma$.

\medskip
{\bf Erratum:}
The main hardness result in the conference version of this paper \cite{Vondrak09} was formulated
for an arbitrary feasibility constraint $\cF$, invariant under $\cG$.
Unfortunately, this was an error and the theorem does not hold in that form
 --- an algorithm could gather some information from querying $\cF$, and combining this
with information obtained by querying the objective function $f$ it could possibly determine the hidden optimal solution.
The possibility of gathering information from the membership oracle for $\cF$ was neglected in the proof. (The reason for this
was probably that the feasibility constraints used in concrete applications of the theorem were very simple and indeed
did not provide any information about the optimum.)
Nevertheless, to correct this issue, one needs to impose a stronger symmetry constraint on $\cF$, namely
the condition that $S \in \cF$ depends only on the symmetrized version of $S$, $\overline{\b1_S} = \E_{\sigma \in \cG}[\b1_{\sigma(S)}]$.
This is the case in all the applications of the hardness theorem in \cite{Vondrak09} and \cite{OV11} and hence
these applications are not affected.

\begin{definition}
\label{def:total-sym}
We call an instance $\max \{f(S): S \in \cF\}$ on a ground set $X$ strongly symmetric with respect to a group of permutations $\cG$ on $X$, if $f(S) = f(\sigma(S))$ for all $S \subseteq X$ and $\sigma \in \cG$, and $S \in \cF \Leftrightarrow S' \in \cF$ whenever $\E_{\sigma \in \cG}[\b1_{\sigma(S)}] = \E_{\sigma \in \cG}[\b1_{\sigma(S')}] $.
\end{definition}

\noindent{\bf Example.} A cardinality constraint, $\cF = \{S \subseteq [n]: |S| \leq k \}$, is strongly symmetric with respect to all permutations, because the condition $S \in \cF$ depends only on the symmetrized vector $\overline{\b1_S} = \frac{|S|}{n} \b1$.
Similarly, a partition matroid constraint, $\cF = \{S \subseteq [n]: |S \cap X_i| \leq k \}$ for disjoint sets $X_i$, is also strongly symmetric.
On the other hand, consider a family of feasible solution $\cF = \{ \{1,2\}, \{2,3\}, \{3,4\}, \{4,1\} \}$. This family is invariant under a group generated by the cyclic rotation $1 \rightarrow 2 \rightarrow 3 \rightarrow 4 \rightarrow 1$. It is not strongly symmetric, because the condition $S \in \cF$ does not depend only the symmetrized vector $\overline{\b1_S} = (\frac14 |S|,\frac14 |S|,\frac14 |S|,\frac14 |S|)$;
some pairs are feasible and others are not.

\

Next, we define the symmetry gap as the ratio between the optimal solution
of $\max \{F(\bx): \bx \in P(\cF)\}$ and the best {\em symmetric} solution
of this problem.

\begin{definition}[Symmetry gap]
Let $\max \{ f(S): S \in {\cal F} \}$ be an instance on a ground set $X$,
which is strongly symmetric with respect to ${\cal G} \subset {\bf S}(X)$.
Let $F(\bx) = \E[f(\hat{\bx})]$ be the multilinear extension of $f(S)$
and $P({\cal F}) = \mbox{conv}(\{\b1_I: I \in {\cal F} \})$ the polytope
associated with $\cal F$. Let $\bar{\bx} = \E_{\sigma \in {\cal G}}[\sigma(\bx)]$.
The symmetry gap of $\max \{ f(S): S \in {\cal F} \}$ is defined as
$\gamma = \overline{OPT} / OPT$ where
$$OPT = \max \{F(\bx): \bx \in P({\cal F})\},$$
$$\overline{OPT} = \max \{F(\bar{\bx}): \bx \in P({\cal F}) \}.$$
\end{definition}

\noindent
We give examples of computing the symmetry gap in Section~\ref{section:hardness-applications}.
Next, we need to define the notion of a {\em refinement} of an instance.
This is a natural way to extend a family of feasible sets to a larger ground set.
In particular, this operation preserves the types of constraints that we care about,
such as cardinality constraints, matroid independence, and matroid base constraints.

\begin{definition}[Refinement]
Let $\cF \subseteq 2^X$, $|X|=k$ and $|N| = n$. We say that
$\tilde{\cF}\subseteq 2^{N \times X}$ is a refinement of $\cal F$, if
$$ \tilde{\cF} = \left\{ \tilde{S} \subseteq N \times X \ \big| \ (x_1,\ldots,x_k) \in P({\cF})
 \mbox{ where } x_j = \frac{1}{n} |\tilde{S} \cap (N \times \{j\})| \right\}. $$
\end{definition}

In other words, in the refined instance, each element $j \in X$ is replaced by a set $N \times \{j\}$.
We call this set the {\em cluster} of elements corresponding to $j$.
A set $\tilde{S}$ is in $\tilde{\cF}$ if and only if the fractions $x_j$ of the respective clusters that are intersected by $\tilde{S}$ form a vector $\bx \in P(\cF)$, i.e. a convex combination of sets in $\cF$.

Our main result is that the symmetry gap for any strongly symmetric instance
translates automatically into hardness of approximation for refined instances.
(See Definition~\ref{def:total-sym} for the notion of being ``strongly symmetric".)
We emphasize that this is a query-complexity lower bound,
and hence independent of assumptions such as $P \neq NP$.

\begin{theorem}
\label{thm:general-hardness}
Let $\max \{ f(S): S \in {\cal F} \}$ be an instance of nonnegative
(optionally monotone) submodular maximization, strongly symmetric with respect to $\cG$,
with symmetry gap $\gamma = \overline{OPT} / OPT$.
Let $\cal C$ be the class of instances $\max \{\tilde{f}(S): S \in \tilde{\cF}\}$
where $\tilde{f}$ is nonnegative (optionally monotone) submodular
and $\tilde{\cF}$ is a refinement of $\cF$.
Then for every $\epsilon > 0$,
any (even randomized) $(1+\epsilon) \gamma$-approximation algorithm for the class $\cal C$ would require
exponentially many value queries to $\tilde{f}(S)$.
\end{theorem}

We remark that the result holds even if the class $\cal C$ is restricted
to instances which are themselves symmetric under a group related to $\cG$
(see the discussion in Section~\ref{section:hardness-proof}, after the proofs of Theorem~\ref{thm:general-hardness}
and \ref{thm:multilinear-hardness}).
On the algorithmic side, submodular maximization seems easier for symmetric instances
and in this case we obtain optimal approximation factors, up to lower-order terms
(see Section~\ref{section:symmetric}).

Our hardness construction yields impossibility results also for solving
the continuous problem $\max \{F(\bx): \bx \in P(\cF) \}$. In the case of matroid constraints,
this is easy to see, because an approximation to the continuous problem gives the same
approximation factor for the discrete problem (by pipage rounding, see Appendix~\ref{app:pipage}).
However, this phenomenon
is more general and we can show that the value of a symmetry gap translates into
an inapproximability result for the multilinear optimization problem under any constraint
satisfying a symmetry condition. 

\begin{theorem}
\label{thm:multilinear-hardness}
Let $\max \{ f(S): S \in {\cal F} \}$ be an instance of nonnegative
(optionally monotone) submodular maximization, strongly symmetric with respect to $\cG$,
with symmetry gap $\gamma = \overline{OPT} / OPT$.
Let $\cal C$ be the class of instances $\max \{\tilde{f}(S): S \in \tilde{\cF}\}$
where $\tilde{f}$ is nonnegative (optionally monotone) submodular
and $\tilde{\cF}$ is a refinement of $\cF$.
Then for every $\epsilon > 0$,
any (even randomized) $(1+\epsilon) \gamma$-approximation algorithm for the multilinear relaxation
$\max \{\tilde{F}(\bx): \bx \in P(\tilde{\cF})\}$ of problems in $\cal C$ would require
exponentially many value queries to $\tilde{f}(S)$.
\end{theorem}

\

\noindent{\bf Additions to the conference version and follow-up work.}
An extended abstract of this work appeared in IEEE FOCS 2009 \cite{Vondrak09}.
As mentioned above, the main theorem in \cite{Vondrak09} suffers from a technical flaw.
This does not affect the applications, but the general theorem in \cite{Vondrak09} is not correct.
We provide a corrected version of the main theorem with a complete proof
(Theorem~\ref{thm:general-hardness}) and we extend this hardness result to the problem of solving
 the multilinear relaxation (Theorem~\ref{thm:multilinear-hardness}).

Subsequently, further work has been done which exploits the symmetry gap concept. 
In \cite{OV11}, it has been proved using Theorem~\ref{thm:general-hardness} that
maximizing a nonnegative submodular function subject to a matroid independence constraint
with a factor better than $0.478$ would require exponentially many queries. Even in the case of a cardinality constraint,
$\max \{f(S): |S| \leq k\}$ cannot be approximated within a factor better than $0.491$ using subexponentially many
queries \cite{OV11}.
In the case of a matroid base constraint, assuming that the fractional base packing number
is $\nu = \frac{k}{k-1}$ for some $k \geq 2$, there is no $(1-e^{-1/k}+\epsilon)$-approximation in the value oracle model \cite{OV11}, improving the hardness of $(1-\frac{1}{\nu}+\epsilon) = (\frac{1}{k}+\epsilon)$-approximation from this paper.
These applications are not affected by the flaw in \cite{Vondrak09},
and they are implied by the corrected version of Theorem~\ref{thm:general-hardness} here.

Recently \cite{DV11}, it has been proved using the symmetry gap technique that combinatorial auctions with submodular bidders do not admit any truthful-in-expectation $1/m^\gamma$-approximation, where $m$ is the number of items and $\gamma>0$ some absolute constant.
This is the first nontrivial hardness result for truthful-in-expectation mechanisms for combinatorial auctions;
it separates the classes of monotone submodular functions and coverage functions,
where a truthful-in-expectation $(1-1/e)$-approximation is possible \cite{DRY11}.
The proof is self-contained and does not formally refer to \cite{Vondrak09}.

Moreover, this hardness result for truthful-in-expectation mechanisms as well as the main hardness result in this paper
have been converted from the oracle setting to a computational complexity setting \cite{DV12a,DV12b}.  This recent work
shows that the hardness of approximation arising from symmetry gap is not limited to instances given by an oracle, but holds
also for instances encoded explicitly on the input, under a suitable complexity-theoretic assumption.

\

\noindent{\bf Organization.}
The rest of the paper is organized as follows.
In Section~\ref{section:hardness-applications}, we present applications
of our main hardness result (Theorem~\ref{thm:general-hardness}) to concrete cases,
in particular we show how it implies the hardness statements in Theorem~\ref{thm:matroid-indep}
and \ref{thm:matroid-bases}.
In Section~\ref{section:hardness-proof}, we present the proofs of Theorem~\ref{thm:general-hardness}
 and Theorem~\ref{thm:multilinear-hardness}.
In Section~\ref{section:algorithms}, we prove the algorithmic results
in Theorem~\ref{thm:matroid-indep} and \ref{thm:matroid-bases}.
In Section~\ref{section:symmetric}, we discuss the special case of symmetric instances.
In the Appendix,
we present a few basic facts concerning submodular functions,
an extension of pipage rounding to matroid independence polytopes (rather than matroid base polytopes),
and other technicalities that would hinder the main exposition.

\section{From symmetry to inapproximability: applications}
\label{section:hardness-applications}

Before we get into the proof of Theorem~\ref{thm:general-hardness},
let us show how it can be applied to a number of specific problems.
Some of these are hardness results that were proved previously by
an ad-hoc method. The last application is a new one
(Theorem~\ref{thm:matroid-bases}).

\

\noindent{\bf Nonmonotone submodular maximization.}
Let $X = \{1,2\}$ and for any $S \subseteq X$,
$f(S) = 1$ if $|S|=1$, and $0$ otherwise.
Consider the instance $\max \{ f(S): S \subseteq X \}$.
In other words, this is the Max Cut problem on the graph $K_2$.
This instance exhibits a simple symmetry, the group of all (two) permutations
on $\{1,2\}$. We get $OPT = F(1,0) = F(0,1) = 1$, while $\overline{OPT} = F(1/2,1/2)
 = 1/2$. Hence, the symmetry gap is $1/2$.

\begin{figure}[here]
\begin{tikzpicture}[scale=.50]

\draw (-10,0) node {};

\filldraw [fill=gray,line width=1mm] (0,0) rectangle (4,2);

\filldraw [fill=white] (0.5,1) .. controls +(0,1) and +(0,1) .. (1.5,1)
 .. controls +(0,-1) and +(0,-1) .. (0.5,1);

\fill (1,1) circle (5pt);
\fill (3,1) circle (5pt);
\draw (1,1) -- (3,1);

\draw (-1,1) node {$X$};

\draw (7,1) node {$\bar{x}_1 = \bar{x}_2 = \frac{1}{2}$};

\end{tikzpicture}

\caption{Symmetric instance for nonmonotone submodular maximization: Max Cut on the graph $K_2$.
The white set denotes the optimal solution, while $\bar{\bx}$ is the (unique) symmetric solution.}

\end{figure}
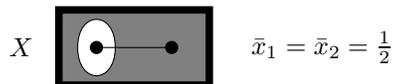

Since $f(S)$ is nonnegative submodular and there is no constraint on $S \subseteq X$,
this will be the case for any refinement of the instance as well.
Theorem~\ref{thm:general-hardness} implies immediately the following: any algorithm achieving
a $(\frac12 + \epsilon)$-approximation for nonnegative (nonmonotone) submodular
maximization requires exponentially many value queries
(which was previously known \cite{FMV07}). Note that a ``trivial instance" implies
a nontrivial hardness result. This is typically the case in applications of Theorem~\ref{thm:general-hardness}.

The same symmetry gap holds if we impose some simple constraints:
the problems $\max \{ f(S): |S| \leq 1 \}$ and $\max \{ f(S): |S| = 1 \}$
have the same symmetry gap as above. Hence, the hardness threshold of $1/2$
also holds for nonmonotone submodular maximization under cardinality constraints
of the type $|S| \leq n/2$, or $|S| = n/2$. This proves the hardness part of
Theorem~\ref{thm:matroid-indep}. This can be derived quite easily
from the construction of \cite{FMV07} as well.

\

\noindent{\bf Monotone submodular maximization.}
Let $X = [k]$ and $f(S) = \min \{|S|, 1\}$.
Consider the instance $\max \{f(S): |S| \leq 1 \}$.
This instance is invariant under all permutations on $[k]$, the symmetric group ${\bf S}_k$.
Note that the instance is {\em strongly symmetric} (Def.~\ref{def:total-sym}) with respect to ${\bf S}_k$,
since the feasibility constraint $|S| \leq 1$ depends only on the symmetrized vector
$\overline{\b1_S} = (\frac{1}{k}|S|,
\ldots, \frac{1}{k}|S|)$. We get $OPT = F(1,0,\ldots,0) = 1$, while
$\overline{OPT} = F(1/k,1/k,\ldots,1/k) = 1 - (1-1/k)^k$.

Here, $f(S)$ is monotone submodular and any refinement of $\cal F$ is
a set system of the type $\tilde{\cF} = \{S: |S| \leq \ell \}$.
Based on our theorem, this implies that any approximation better than $1 - (1-1/k)^k$
for monotone submodular maximization subject
to a cardinality constraint would require exponentially many value queries.
Since this holds for any fixed $k$, we get the same hardness
result for any $\beta > \lim_{k \rightarrow \infty} (1 - (1-1/k)^k) = 1-1/e$
(which was previously known \cite{NW78}).

\

\noindent{\bf Submodular maximization over matroid bases.}
Let $X = A \cup B$, $A = \{a_1, \ldots, a_k\}$, $B = \{b_1, \ldots, b_k\}$
and ${\cal F} = \{ S: |S \cap A| = 1 \ \& \ |S \cap B| = k-1 \}$.
We define $f(S) = \sum_{i=1}^{k} f_i(S)$ where
$f_i(S) = 1$ if $a_i \in S \ \& \ b_i \notin S$, and $0$ otherwise.
This instance can be viewed as a Maximum Directed Cut problem on a graph
of $k$ disjoint arcs, under the constraint that exactly one arc tail
and $k-1$ arc heads should be on the left-hand side ($S$).
An optimal solution is for example $S = \{a_1, b_2, b_3, \ldots, b_k \}$,
which gives $OPT = 1$.
The symmetry here is that we can apply the same permutation to $A$ and $B$
simultaneously. Again, the feasibility of a set $S$ depends only on the symmetrized vector $\overline{\b1_S}$:
in fact $S \in \cF$ if and only if $\overline{\b1_S} = (\frac1k,\ldots,\frac1k,1-\frac1k,\ldots,1-\frac1k)$.
There is a unique symmetric solution
$\bar{\bx} = (\frac1k,\ldots,\frac1k,1-\frac1k,\ldots,1-\frac1k)$, and
$\overline{OPT} = F(\bar{\bx}) = \E[f(\hat{\bar{\bx}})] = \sum_{i=1}^{k}
 \E[f_i(\hat{\bar{\bx}})] = \frac1k$
(since each arc appears in the directed cut induced by $\hat{\bar{\bx}}$ with probability $\frac{1}{k^2}$).

\begin{figure}[here]
\begin{tikzpicture}[scale=.60]

\draw (-6,0) node {};

\filldraw [fill=gray,line width=1mm] (0,0) rectangle (9,2);
\filldraw [fill=gray,line width=1mm] (0,2) rectangle (9,4);

\filldraw [fill=white] (0.5,3) .. controls +(0,1) and +(0,1) .. (1.5,3)
 .. controls +(0,-1) and +(0,-1) .. (0.5,3);
\filldraw [fill=white] (1.5,1) .. controls +(0,1) and +(0,1) .. (8.5,1)
 .. controls +(0,-1) and +(0,-1) .. (1.5,1);

\fill (1,1) circle (5pt);
\fill (2,1) circle (5pt);
\fill (3,1) circle (5pt);
\fill (4,1) circle (5pt);
\fill (5,1) circle (5pt);
\fill (6,1) circle (5pt);
\fill (7,1) circle (5pt);
\fill (8,1) circle (5pt);
\fill (1,3) circle (5pt);
\fill (2,3) circle (5pt);
\fill (3,3) circle (5pt);
\fill (4,3) circle (5pt);
\fill (5,3) circle (5pt);
\fill (6,3) circle (5pt);
\fill (7,3) circle (5pt);
\fill (8,3) circle (5pt);
\draw[-latex] [line width=0.5mm] (1,3) -- (1,1);
\draw[-latex] [line width=0.5mm] (2,3) -- (2,1);
\draw[-latex] [line width=0.5mm] (3,3) -- (3,1);
\draw[-latex] [line width=0.5mm] (4,3) -- (4,1);
\draw[-latex] [line width=0.5mm] (5,3) -- (5,1);
\draw[-latex] [line width=0.5mm] (6,3) -- (6,1);
\draw[-latex] [line width=0.5mm] (7,3) -- (7,1);
\draw[-latex] [line width=0.5mm] (8,3) -- (8,1);

\draw (-1,1) node {$B$};
\draw (-1,3) node {$A$};

\draw (11,3) node {$\bar{x}_{a_i} = \frac{1}{k}$};
\draw (11,1) node {$\bar{x}_{b_i} = 1 - \frac{1}{k}$};

\end{tikzpicture}

\caption{Symmetric instance for submodular maximization over matroid bases.}

\end{figure}
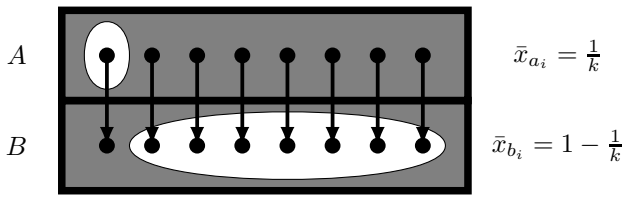
The refined instances are instances of (nonmonotone) submodular maximization
over the bases of a matroid, where the ground set is partitioned into $A \cup B$
and we should take a $\frac1k$-fraction of $A$ and a $(1-\frac1k)$-fraction of $B$.
(This means that the fractional packing number of bases is $\nu = \frac{k}{k-1}$.)
Our theorem implies that for this class of instances, an approximation better
than $1/k$ is impossible - this proves the hardness part of Theorem~\ref{thm:matroid-bases}.

\

Observe that in all the cases mentioned above, the multilinear relaxation is equivalent
to the original problem, in the sense that any fractional solution can be rounded 
without any loss in the objective value. This implies that the same hardness factors apply to
solving the multilinear relaxation of the respective problems. In particular, using the last result
(for matroid bases), we obtain that the multilinear optimization problem $\max \{ F(\bx): \bx \in P \}$
does not admit a constant factor for nonnegative submodular functions and matroid base polytopes.
(We remark that a $(1-1/e)$-approximation can be achieved for any {\em monotone} submodular function
and any solvable polytope, i.e.~polytope over which we can optimize linear functions \cite{Vondrak08}.)

As Theorem~\ref{thm:multilinear-hardness} shows, this holds more generally - any symmetry gap
construction gives an inapproximability result for solving the multilinear optimization problem
$\max \{F(\bx): \bx \in P\}$. This in fact implies limits on what hardness results we can
possibly hope for using this technique. For instance, we cannot prove using the symmetry gap
that the monotone submodular maximization problem subject to the intersection of $k$ matroid constraints
does not admit a constant factor  - because we would also prove that the respective multilinear relaxation
does not admit such an approximation. But we know from \cite{Vondrak08} that a $(1-1/e)$-approximation
is possible for the multilinear problem in this case.

Hence, the hardness arising from the symmetry gap is related to the difficulty of solving
the multilinear optimization problem rather than the difficulty of rounding a fractional solution.
Thus this technique is primarily suited to optimization problems where the multilinear optimization problem
captures closely the original discrete problem.

\section{From symmetry to inapproximability: proof}
\label{section:hardness-proof}

\paragraph{The roadmap}
At a high level,
our proof resembles the constructions of \cite{FMV07,MSV08}.
We construct instances based on continuous functions $F(\bx)$, $G(\bx)$,
whose optima differ by a gap 
for which we want to prove hardness. Then we show that after a certain perturbation,
the two instances are very hard to distinguish.
This paper generalizes the ideas of \cite{FMV07,MSV08} and brings
two new ingredients. First, we show that the functions
$F(\bx), G(\bx)$, which are ``pulled out of the hat'' in \cite{FMV07,MSV08},
can be produced in a natural way from the multilinear relaxation
of the respective problem, using the notion of a {\em symmetry gap}.
Secondly, the functions $F(\bx), G(\bx)$ are perturbed in a way that makes
them indistinguishable and this forms the main technical part
of the proof. In \cite{FMV07}, this step is quite simple.
In \cite{MSV08}, the perturbation is more complicated, but still relies
on properties of the functions $F(\bx), G(\bx)$ specific to that application.
The construction that we present here (Lemma~\ref{lemma:final-fix})
uses the symmetry properties of a fixed instance in a generic fashion.

\

First, let us present an outline of our construction. Given an instance
$\max \{f(S): S \in {\cal F}\}$ exhibiting a symmetry gap $\gamma$,
we consider two smooth submodular\footnote{"Smooth submodularity"
means the condition $\mixdiff{F}{x_i}{x_j} \leq 0$ for all $i,j$.}
functions, $F(\bx)$ and $G(\bx)$.
The first one is the multilinear extension $F(\bx) = \E[f(\hat{\bx})]$,
while the second one is its symmetrized version $G(\bx) = F(\bar{\bx})$.
We modify these functions slightly so that we obtain
functions $\hat{F}(\bx)$ and $\hat{G}(\bx)$ with the following property:
For any vector $\bx$ which is close to its symmetrized version
$\bar{\bx}$, $\hat{F}(\bx) = \hat{G}(\bx)$.
The functions $\hat{F}(\bx), \hat{G}(\bx)$ induce instances of
submodular maximization on the refined ground sets. The way we define
discrete instances based on $\hat{F}(\bx), \hat{G}(\bx)$ is natural,
using the following lemma. 
Essentially, we interpret the fractional variables as fractions of clusters
in the refined instance.

\begin{lemma}
\label{lemma:smooth-submodular}
Let $F:[0,1]^X \rightarrow \RR$,
 $N = [n]$, $n \geq 1$, and define $f:2^{N \times X} \rightarrow \RR$
so that $f(S) = F(\bx)$ where $x_i = \frac{1}{n} |S \cap (N \times \{i\})|$. Then
\begin{enumerate}
\item If $\partdiff{F}{x_i} \geq 0$ everywhere for each $i$, then
$f$ is monotone.
\item If  the first partial derivatives of $F$ are absolutely continuous\footnote{
A function $F:[0,1]^X \rightarrow \RR$ is absolutely continuous,
if $\forall \epsilon>0; \exists \delta>0; \sum_{i=1}^{t} ||\bx_i-\by_i|| < \delta \Rightarrow
\sum_{i=1}^{t} |F(\bx_i) - F(\by_i)| < \epsilon$.} and $\mixdiff{F}{x_i}{x_j} \leq 0$ almost everywhere for all $i,j$,
then $f$ is submodular.
\end{enumerate}
\end{lemma}

\begin{proof}
First, assume $\partdiff{F}{x_i} \geq 0$ everywhere for all $i$. This implies that $F$ is nondecreasing in every coordinate,
i.e. $F(\bx) \leq F(\by)$ whenever $\bx \leq \by$. This means that $f(S) \leq f(T)$ whenever $S \subseteq T$.

Next, assume $\partdiff{F}{x_i}$ is absolutely continuous for each $i$ and $\mixdiff{F}{x_j}{x_i} \leq 0$ almost everywhere for all $i,j$.
We want to prove that $\partdiff{F}{x_i} |_\bx \geq \partdiff{F}{x_i} |_\by$ whenever $\bx \leq \by$, which implies that the marginal
values of $f$ are nonincreasing.

Let $\bx \leq \by$,
fix $\delta>0$ arbitrarily small, and pick $\bx',\by'$ such that $||\bx'-\bx||<\delta, ||\by'-\by||<\delta, \bx' \leq \by'$ and
on the line segment $[\bx', \by']$, we have $\mixdiff{F}{x_j}{x_i} \leq 0$ except for a set of (1-dimensional) measure zero. If such a pair of points
$\bx', \by'$ does not exist, it means that there are sets $A,B$ of positive measure such that
$\bx \in A, \by \in B$ and for any
$\bx' \in A, \by' \in B$, the line segment $[\bx',\by']$ contains a subset of positive (1-dimensional) measure where $\mixdiff{F}{x_j}{x_i}$ 
for some $j$ is positive or undefined. This would imply that $[0,1]^X$ contains a subset of positive measure where
$\mixdiff{F}{x_j}{x_i}$ for some $j$ is positive or undefined, which we assume is not the case.

Therefore, there is a pair of points $\bx', \by'$ as described above.
We compare $\partdiff{F}{x_i} |_{\bx'}$ and $\partdiff{F}{x_i} |_{\by'}$
by integrating along the line segment $[\bx', \by']$. Since $\partdiff{F}{x_i}$ is absolutely continuous and $\partdiff{}{x_j} \partdiff{F}{x_i} = \mixdiff{F}{x_j}{x_i} \leq 0$ along this line segment for all $j$ except for a set of measure zero, we obtain $\partdiff{F}{x_i} |_{\bx'} \geq \partdiff{F}{x_i} |_{\by'}$. This is true for $\bx', \by'$ arbitrarily
close to $\bx, \by$, and hence by continuity of $\partdiff{F}{x_i}$, we get $\partdiff{F}{x_i} |_{\bx} \geq \partdiff{F}{x_i} |_{\by}$.
This implies that the marginal values of $f$ are nonincreasing.
\end{proof}

The way we construct $\hat{F}(\bx), \hat{G}(\bx)$ is such that,
given a large enough refinement of the ground set,
it is impossible to distinguish the instances corresponding to
$\hat{F}(\bx)$ and $\hat{G}(\bx)$. As we argue more precisely later,
this holds because if the ground set is large and labeled in a random way
(considering the symmetry group of the instance), a query about a vector $\bx$
effectively becomes a query about the symmetrized vector $\bar{\bx}$.
We would like this property to imply that all queries with high probability fall
in the region where $\hat{F}(\bx) = \hat{G}(\bx)$ and the inability to distinguish between $\hat{F}$
and $\hat{G}$ gives the hardness result that we want.
The following lemma gives the precise properties of $\hat{F}(\bx)$
and $\hat{G}(\bx)$ that we need.

\begin{lemma}
\label{lemma:final-fix}
Consider a function $f:2^X \rightarrow \RR_+$ invariant under a group
of permutations $\cal G$ on the ground set $X$.
Let $F(\bx) = \E[f(\hat{\bx})]$, $\bar{\bx} = \E_{\sigma \in \cG}[\sigma(\bx)]$,
and fix any $\epsilon > 0$.
Then there is $\delta > 0$ and functions $\hat{F}, \hat{G}:[0,1]^X \rightarrow \RR_+$
(which are also symmetric with respect to $\cG$), satisfying:
\begin{enumerate}
\item For all $\bx \in [0,1]^X$, $\hat{G}(\bx) = \hat{F}(\bar{\bx})$.
\item For all $\bx \in [0,1]^X$, $|\hat{F}(\bx) - F(\bx)| \leq \epsilon$.
\item Whenever $||\bx - \bar{\bx}||^2 \leq \delta$, $\hat{F}(\bx) = \hat{G}(\bx)$ 
and the value depends only on $\bar{\bx}$.
\item The first partial derivatives of $\hat{F}, \hat{G}$ are absolutely continuous. 
\item If $f$ is monotone, then $\partdiff{\hat{F}}{x_i} \geq 0$ and
   $\partdiff{\hat{G}}{x_i} \geq 0$ everywhere.
\item If $f$ is submodular, then $\mixdiff{\hat{F}}{x_i}{x_j} \leq 0$ and
   $\mixdiff{\hat{G}}{x_i}{x_j} \leq 0$ almost everywhere.
\end{enumerate}
\end{lemma}

The proof of this lemma is the main technical part of this paper
and we defer it to the end of this section. Assuming this lemma,
we first finish the proof of the main theorem. We prove the following.

\begin{lemma}
\label{lemma:indistinguish}
Let $\hat{F}, \hat{G}$ be the two functions provided by Lemma~\ref{lemma:final-fix}.
For a parameter $n \in \ZZ_+$ and $N = [n]$, define two discrete functions
$\hat{f}, \hat{g}: 2^{N \times X} \rightarrow \RR_+$ as follows:
Let $\sigma^{(i)}$ be an arbitrary permutation in $\cG$ for each $i \in N$.
For every set $S \subseteq N \times X$, we define a vector $\xi(S) \in [0,1]^X$ by
$$ \xi_j(S) = \frac{1}{n} \left|\{i \in N: (i,\sigma^{(i)}(j)) \in S \}\right|.$$
Let us define:
$$ \hat{f}(S) = \hat{F}(\xi(S)), \ \ \ \ \ \hat{g}(S) = \hat{G}(\xi(S)).$$
In addition, let $\tilde{\cF} = \{\tilde{S}: \xi(\tilde{S}) \in P(\cF)\}$
be a feasibility constraint such that the condition $S \in \cF$ depends
only on the symmetrized vector $\overline{\b1_S}$.
Then deciding whether an instance given by value/membership oracles is
$\max \{\hat{f}(S): S \in \tilde{\cF}\}$ or $\max \{\hat{g}(S): S \in \tilde{\cF} \}$
(even by a randomized algorithm, with any constant probability of success)
requires an exponential number of queries.
\end{lemma}

\begin{proof}
Let $\sigma^{(i)} \in \cG$ be chosen independently at random for each $i \in N$ and consider the instances
$\max\{\hat{f}(S): S \in \tilde{\cF} \}$, $\max \{\hat{g}(S): S \in \tilde{\cF} \}$ as described in the lemma. We show that
every deterministic algorithm will follow the same computation path and return the same answer on both instances,
with high probability. By Yao's principle, this means that every randomized algorithm returns
the same answer for the two instances with high probability, for some particular $\sigma^{(i)} \in \cG$.

The feasible sets in the refined instance, $S \in \tilde{\cF}$, are such that
the respective vector $\xi(S)$ is in the polytope $P({\cF})$.
Since the instance is strongly symmetric, the condition $S \in \cF$ depends
only on the symmetrized vector $\overline{\b1_S}$. Hence, the condition $\xi(S) \in P(\cF)$
depends only on the symmetrized vector $\overline{\xi(S)}$. Therefore, $S \in \tilde{\cF} \Leftrightarrow
\xi(S) \in P(\cF) \Leftrightarrow \overline{\xi(S)} \in P(\cF)$.
We have $\overline{\xi(S)}_j = \E_{\sigma \in \cG}[\frac{1}{n} \left|\{i \in N: (i,\sigma^{(i)}(\sigma(j))) \in S \}\right|].$ The distribution of $\sigma^{(i)} \circ \sigma$ is the same as that of $\sigma$, i.e.~uniform over $\cG$. Hence, $\overline{\xi(S)}_j$ and consequently the condition $S \in \tilde{\cF}$ does not depend on $\sigma^{(i)}$ in any way.
Intuitively, an algorithm cannot learn any information about the permutations $\sigma^{(i)}$ by querying the feasibility oracle,
since the feasibility condition $S \in \tilde{\cF}$ does not depend on $\sigma^{(i)}$ for any $i \in N$.

The main part of the proof is to show that even queries to the objective function are unlikely to reveal
any information about the permutations $\sigma^{(i)}$.
The key observation is that for any fixed query $Q$ to the objective function,
the associated vector $\bq = \xi(Q)$ is very likely to be close
to its symmetrized version $\bar{\bq}$.
To see this, consider a query $Q$. The associated vector $\bq = \xi(Q)$
is determined by
$$ q_j = \frac{1}{n} |\{i \in N: (i,\sigma^{(i)}(j)) \in Q\}|
 = \frac{1}{n} \sum_{i=1}^{n} Q_{ij} $$
where $Q_{ij}$ is an indicator variable of the event $(i,\sigma^{(i)}(j)) \in Q$.
This is a random event due to the randomness in $\sigma^{(i)}$.
We have
$$ \E[Q_{ij}] = \Pr[Q_{ij}=1] = \Pr_{\sigma^{(i)} \in \cG}[(i,\sigma^{(i)}(j)) \in Q].$$
Adding up these expectations over $i \in N$, we get
$$ \sum_{i \in N} \E[Q_{ij}]
 = \sum_{i \in N} \Pr_{\sigma^{(i)} \in \cG}[(i,\sigma^{(i)}(j)) \in Q] 
 = \E_{\sigma \in \cG} [|\{i \in N: (i,\sigma(j)) \in Q \}|].$$
For the purposes of expectation, the independence of $\sigma^{(1)},
 \ldots, \sigma^{(n)}$ is irrelevant and that is why we replace them by one random permutation $\sigma$.
On the other hand, consider the symmetrized vector $\bar{\bq}$, with coordinates
$$ \bar{q}_j = \E_{\sigma \in \cG}[q_{\sigma(j)}]
 = \frac{1}{n} \E_{\sigma \in \cG}[|\{i \in N: (i,\sigma^{(i)}(\sigma(j))) \in Q\}|]
 = \frac{1}{n} \E_{\sigma \in \cG}[|\{i \in N: (i,\sigma(j)) \in Q \}|] $$
using the fact that the distribution of $\sigma^{(i)} \circ \sigma$
is the same as the distribution of $\sigma$ - uniformly random over $\cG$.
Note that the vector $\bq$ depends on the random permutations $\sigma^{(i)}$
but the symmetrized vector $\bar{\bq}$ does not; this will be also useful
in the following. For now, we summarize that
$$ \bar{q}_j = \frac{1}{n} \sum_{i=1}^{n} \E[Q_{ij}] = \E[q_j].$$
Since each permutation $\sigma^{(i)}$ is chosen independently, the random variables
$\{Q_{ij}: 1 \leq i \leq n\}$ are independent (for a fixed $j$).
We can apply Chernoff's bound (see e.g. \cite{AlonSpencer}, Corollary A.1.7):
$$ \Pr\left[\left|\sum_{i=1}^{n} Q_{ij}- \sum_{i=1}^{n} \E[Q_{ij}] \right| > a \right]
 < 2e^{-2a^2 / n}.$$
Using $q_j = \frac{1}{n} \sum_{i=1}^{n} Q_{ij}$,
 $\bar{q}_j = \frac{1}{n} \sum_{i=1}^{n} \E[Q_{ij}]$
and setting $a = n \sqrt{\delta/|X|}$, we obtain
$$ \Pr\left[|q_j - \bar{q}_j| > \sqrt{{\delta}/{|X|}}\right]
 < 2e^{-2 n \delta / |X|}.$$
By the union bound,
\begin{equation}
\label{eq:D(q)}
  \Pr[||\bq-\bar{\bq}||^2 > \delta] \leq \sum_{j \in X} \Pr[|q_j-\bar{q}_j|^2 > {\delta}/{|X|}] < 2|X| e^{-2n \delta/|X|}.
\end{equation}
Note that while $\delta$ and $|X|$ are constants, $n$ grows as the size
of the refinement and hence the probability is exponentially small in the
size of the ground set $N \times X$.

Define $D(\bq) = ||\bq-\bar{\bq}||^2$. As long as $D(\bq) \leq \delta$ for every query issued by the algorithm,
the answers do not depend on the randomness of the input. This is because
then the values of $\hat{F}(\bq)$ and $\hat{G}(\bq)$ depend only on $\bar{\bq}$,
which is independent of the random permutations $\sigma^{(i)}$,
as we argued above.
Therefore, assuming that $D(\bq) \leq \delta$ for each query,
the algorithm will always follow the same path
of computation and issue the same sequence of queries $\cal S$.
(Note that this is just a fixed sequence which can be written down
before we started running the algorithm.)
Assume that $|{\cal S}| = e^{o(n)}$, i.e.,~the number of queries is subexponential in $n$.
By (\ref{eq:D(q)}), using a union bound over all $Q \in {\cal S}$,
it happens with probability $1 - e^{-\Omega(n)}$ that $\DD(\bq) = ||\bq-\bar{\bq}||^2 \leq \delta$ for all $Q \in {\cal S}$.
(Note that $\delta$ and $|X|$ are constants here.)
Then, the algorithm indeed follows this path of computation and gives the same answer.
In particular, the answer does not depend on whether the instance is $\max \{\hat{f}(S): S \in \tilde{\cF}\}$
or $\max \{\hat{g}(S): S \in \tilde{\cF}\}$.
\end{proof}

\begin{proof}[Proof of Theorem~\ref{thm:general-hardness}]
Fix an $\epsilon > 0$.
Given an instance $\max \{f(S): S \in \cF\}$
strongly symmetric under $\cG$, let $\hat{F}, \hat{G}: [0,1]^X \rightarrow \RR$
be the two functions provided by Lemma~\ref{lemma:final-fix}.
We choose a large number $n$ and consider a refinement $\tilde{\cF}$
on the ground set $N \times X$, where $N = [n]$.
We define discrete instances of submodular maximization
$\max \{ \hat{f}(S): S \in \tilde{\cF} \}$ and $\max \{ \hat{g}(S): S \in \tilde{\cF} \}$.
As in Lemma~\ref{lemma:indistinguish}, for each $i \in N$
we choose a random permutation $\sigma^{(i)} \in \cG$.
This can be viewed as a random shuffle of the labeling of the ground set
before we present it to an algorithm.
For every set $S \subseteq N \times X$, we define a vector $\xi(S) \in [0,1]^X$ by
$$ \xi_j(S) = \frac{1}{n} \left|\{i \in N: (i,\sigma^{(i)}(j)) \in S \}\right|.$$
In other words, $\xi_j(S)$ measures the fraction of copies of element $j$
contained in $S$; however, for each $i$ the $i$-copies of all elements
are shuffled by $\sigma^{(i)}$. Next, we define
$$ \hat{f}(S) = \hat{F}(\xi(S)), \ \ \ \ \ \hat{g}(S) = \hat{G}(\xi(S)).$$
We claim that $\hat{f}$ and $\hat{g}$ are submodular (for any fixed $\xi$).
Note that the effect
of $\sigma^{(i)}$ is just a renaming (or shuffling) of the elements
of $N \times X$, and hence for the purpose of proving submodularity we can assume
that $\sigma^{(i)}$ is the identity for all $i$. Then, $\xi_j(S) = \frac{1}{n}
 |S \cap (N \times \{j\})|$. Due to Lemma~\ref{lemma:smooth-submodular},
the property $\mixdiff{\hat{F}}{x_i}{x_j} \leq 0$ (almost everywhere) implies that $\hat{f}$ is submodular.
In addition, if the original instance was monotone, then $\partdiff{\hat{F}}{x_j} \geq 0$
and $\hat{f}$ is monotone. The same holds for $\hat{g}$.

The value of $\hat{g}(S)$ for any feasible solution $S \in \tilde{\cF}$
is bounded by $\hat{g}(S) = \hat{G}(\xi(S)) = \hat{F}(\overline{\xi(S)})
\leq \overline{OPT} + \epsilon$.
On the other hand, let $\bx^*$ denote a point where the optimum of the continuous problem
$\max \{\hat{F}(\bx): \bx \in P({\cF}) \}$ is attained, i.e. $\hat{F}(\bx^*) \geq OPT - \epsilon$.
For a large enough $n$, we can approximate the point $\bx^*$ arbitrarily closely
by a rational vector with $n$ in the denominator,
which corresponds to a discrete solution $S^* \in \tilde{\cF}$ whose
value $\hat{f}(S^*)$ is at least, say, $OPT - 2 \epsilon$.
Hence, the ratio between the optima of the {\em discrete} optimization
problems $\max \{\hat{f}(S): S \in \tilde{\cF}\}$ and $\max \{\hat{g}(S): S \in \tilde{\cF} \}$
can be made at most $\frac{\overline{OPT} + \epsilon}{OPT - 2 \epsilon}$, i.e. arbitrarily close
to the symmetry gap $\gamma = \frac{\overline{OPT}}{OPT}$.

By Lemma~\ref{lemma:indistinguish}, distinguishing the two instances
$\max \{\hat{f}(S): S \in \tilde{\cF} \}$  and $\max \{\hat{g}(S): S \in \tilde{\cF} \}$,
even by a randomized algorithm, requires an exponential number of value queries.
Therefore, we cannot estimate the optimum within a factor better than $\gamma$.
\end{proof}

Next, we prove Theorem~\ref{thm:multilinear-hardness} (again assuming Lemma~\ref{lemma:final-fix}),
i.e.~an analogous hardness result for solving the multilinear optimization problem.

\begin{proof}[Proof of Theorem~\ref{thm:multilinear-hardness}]
Given a symmetric instance $\max \{f(S): S \in \cF\}$ and $\epsilon>0$, we construct
refined and modified instances $\max \{\hat{f}(S): S \in \tilde{\cF} \}$,
 $\max \{\hat{g}(S): S \in \tilde{\cF} \}$,
derived from the continuous functions $\hat{F}, \hat{G}$ provided by Lemma~\ref{lemma:final-fix},
exactly as we did in the proof of Theorem~\ref{thm:general-hardness}.
Lemma~\ref{lemma:indistinguish} states that these two instances cannot be distinguished using a subexponential number of value queries. Furthermore, the gap between the two modified instances corresponds to the symmetry gap $\gamma$ of the original instance: $\max \{\hat{f}(S): S \in \tilde{\cF} \} \geq OPT - 2 \epsilon$ and $\max \{\hat{g}(S): S \in \tilde{\cF} \} \leq \overline{OPT} + \epsilon = \gamma OPT + \epsilon$.

Now we consider the multilinear relaxations of the two refined instances, $\max \{\check{F}(\bx): \bx \in P(\tilde{\cF})\}$
and $\max \{\check{G}(\bx): \bx \in P(\tilde{\cF})\}$. Note that $\check{F}, \check{G}$, (although related to $\hat{F},
\hat{G}$) are not exactly the same as the functions $\hat{F}, \hat{G}$; in particular, they are defined
on a larger (refined) domain. However, we show that the gap between the optima of the two instances remains
the same.

First, the value of $\max \{\check{F}(\bx): \bx \in P(\cF)\}$ is at least the optimum of the discrete problem,
$\max \{\hat{f}(S): S \in \tilde{\cF} \}$, which is at least $OPT - 2\epsilon$ as in the proof of
Theorem~\ref{thm:general-hardness}. The value of $\max \{\check{G}(\bx): \bx \in P(\cF)\}$ can be analyzed as follows. For any fractional solution $\bx \in P(\cF)$, the value of $\check{G}(\bx)$
is the expectation $\E[\hat{g}(\hat{\bx})]$, where $\hat{\bx}$ is obtained by independently rounding the
coordinates of $\bx$ to $\{0,1\}$. Recall that $\hat{g}$ is obtained by discretizing the continuous function
$\hat{G}$ (using Lemma~\ref{lemma:smooth-submodular}). In particular, $\hat{g}(S) = \hat{G}(\tilde{\bx})$ where
$\tilde{x}_i = \frac{1}{n} |S \cap (N \times \{i\})|$ is the fraction of the respective cluster contained in $S$, and $|N| = n$ is the size of each cluster (the refinement parameter). If $\b1_S = \hat{\bx}$, i.e. $S$ is chosen by independent sampling with probabilities according to $\bx$, then for large $n$ the fractions $\frac{1}{n}|S \cap (N \times \{i\})|$ will be strongly concentrated around their expectation. As $\hat{G}$ is continuous, we get
$\lim_{n \rightarrow \infty} \E[\hat{g}(\hat{\bx})] = \lim_{n \rightarrow \infty} \E[\hat{G}(\tilde{\bx})]
 = \hat{G}(\E[\tilde{\bx}]) = \hat{G}(\bar{\bx})$.
Here, $\bar{\bx}$ is the vector $\bx$ projected back to the original ground set $X$, where the coordinates
of each cluster have been averaged. By construction of the refinement, if $\bx \in P(\tilde{\cF})$ then
$\bar{\bx}$ is in the polytope corresponding to the original instance, $P(\cF)$.
Therefore, $\hat{G}(\bar{\bx}) \leq
\max \{\hat{G}(\bx): \bx \in P(\cF) \} \leq \gamma OPT + \epsilon$. For large enough $n$, this means that
$\max \{\check{G}(\bx): \bx \in P(\tilde{\cF}) \} \leq \gamma OPT + 2 \epsilon$. This holds for
an arbitrarily small fixed $\epsilon>0$, and hence the gap between the instances 
$\max \{\check{F}(\bx): \bx \in P(\tilde{\cF})\}$ and $\max \{\check{G}(\bx): \bx \in P(\tilde{\cF})\}$
(which cannot be distinguished) can be made arbitrarily close to $\gamma$.
\end{proof}

\paragraph{Hardness for symmetric instances}
We remark that since Lemma~\ref{lemma:final-fix} provides functions $\hat{F}$ and $\hat{G}$ symmetric
under $\cG$, the refined instances that we define are invariant with respect to
the following symmetries: permute the copies of each element in an arbitrary
way, and permute the classes of copies according to any permutation
$\sigma \in \cG$. This means that our hardness results also hold
for instances satisfying such symmetry properties.

\

It remains to prove Lemma~\ref{lemma:final-fix}.
Before we move to the final construction of $\hat{F}(\bx)$ and $\hat{G}(\bx)$,
we construct as an intermediate step a function $\tilde{F}(\bx)$ which is helpful
in the analysis.

\

\paragraph{Construction}
Let us construct a function $\tilde{F}(\bx)$ which satisfies the following:
\begin{itemize}
\item For $\bx$ ``sufficiently close'' to $\bar{\bx}$, $\tilde{F}(\bx) = G(\bx)$.
\item For $\bx$ ``sufficiently far away'' from $\bar{\bx}$, $\tilde{F}(\bx) \simeq F(\bx)$.
\item The function $\tilde{F}(\bx)$ is ``approximately" smooth submodular.
\end{itemize}
Once we have $\tilde{F}(\bx)$, we can fix it to obtain a smooth submodular
function $\hat{F}(\bx)$, which is still close to the original function $F(\bx)$.
We also fix $G(\bx)$ in the same way, to obtain a function $\hat{G}(\bx)$
which is equal to $\hat{F}(\bx)$ whenever $\bx$ is close to $\bar{\bx}$. We defer
this step until the end.

We define $\tilde{F}(\bx)$ as a convex linear combination
of $F(\bx)$ and $G(\bx)$, guided by a ``smooth transition'' function, depending
on the distance of $\bx$ from $\bar{\bx}$. The form that we use is the following:\footnote{We remark
that a construction analogous to \cite{MSV08}
would be $\tilde{F}(\bx) = F(\bx) - \phi(H(\bx))$ where $H(\bx) = F(\bx) - G(\bx)$.
While this makes the analysis easier in \cite{MSV08}, it cannot be used in general.
Roughly speaking, the problem is that in general the partial derivatives of $H(\bx)$
are not bounded in any way by the value of $H(\bx)$.}
$$ \tilde{F}(\bx) = (1-\phi(D(\bx))) F(\bx) + \phi(D(\bx)) G(\bx) $$
where $\phi:\RR_+ \rightarrow [0,1]$ is a suitable smooth function,
and
$$ D(\bx) = ||\bx - \bar{\bx}||^2 = \sum_i (x_i - \bar{x}_i)^2.$$
The idea is that when $\bx$ is close to $\bar{\bx}$, $\phi(D(\bx))$ should be close to $1$,
i.e. the convex linear combination should give most of the weight to $G(\bx)$.
The weight should shift gradually to $F(\bx)$ as $\bx$ gets further away from $\bar{\bx}$.
Therefore, we define $\phi(t) = 1$ in a small interval $t \in [0,\delta]$,
and $\phi(t)$ tends to $0$ as $t$ increases.
This guarantees that $\tilde{F}(\bx) = G(\bx)$ whenever
$D(\bx) = ||\bx - \bar{\bx}||^2 \leq \delta$.
We defer the precise construction of $\phi(t)$ to Lemma~\ref{lemma:phi-construction},
after we determine what properties we need from $\phi(t)$.
Note that regardless of the definition of $\phi(t)$, $\tilde{F}(\bx)$
is symmetric with respect to $\cG$, since $F(\bx), G(\bx)$ and $D(\bx)$ are.

\

\paragraph{Analysis of the construction}
Due to the construction of $\tilde{F}(\bx)$, it is clear that when
$D(\bx) = ||\bx-\bar{\bx}||^2$ is small, $\tilde{F}(\bx) = G(\bx)$.
When $D(\bx)$ is large, $\tilde{F}(\bx) \simeq F(\bx)$.
The main issue, however, is whether we can say something about the first
and second partial derivatives of $\tilde{F}$. This is crucial for
the properties of monotonicity and submodularity, which we would like to preserve.
Let us write $\tilde{F}(\bx)$ as
$$ \tilde{F}(\bx) = F(\bx) - \phi(D(\bx)) H(\bx) $$
where $H(\bx) = F(\bx) - G(\bx)$. By differentiating once, we get
\begin{equation}
\label{eq:partdiff}
\partdiff{\tilde{F}}{x_i} = \partdiff{F}{x_i} - \phi(D(\bx)) \partdiff{H}{x_i}
 - \phi'(D(\bx)) \partdiff{D}{x_i} H(\bx)
\end{equation}
and by differentiating twice,
\begin{eqnarray}
\label{eq:mixdiff}
 \mixdiff{\tilde{F}}{x_i}{x_j} & = & \mixdiff{F}{x_i}{x_j}
- \phi(\DD(\bx)) \mixdiff{H}{x_i}{x_j}
 - \phi''(\DD(\bx)) \partdiff{\DD}{x_i} \partdiff{\DD}{x_j} H(\bx) \\
 & & - \phi'(\DD(\bx)) \left( \partdiff{\DD}{x_j} \partdiff{H}{x_i}
  + \mixdiff{\DD}{x_i}{x_j} H(\bx) + \partdiff{\DD}{x_i} \partdiff{H}{x_j} \right) \nonumber.
\end{eqnarray}

The first two terms on the right-hand sides of (\ref{eq:partdiff}) and
(\ref{eq:mixdiff}) are not bothering us, because they form
convex linear combinations of the derivatives of $F(\bx)$ and $G(\bx)$,
which have the properties that we need. The remaining terms might cause
problems, however, and we need to estimate them.

Our strategy is to define $\phi(t)$ in such a way that it eliminates
the influence of partial derivatives of $D$ and $H$ where they become too large.
Roughly speaking, $D$ and $H$ have negligible partial derivatives when $\bx$
is very close to $\bar{\bx}$. As $\bx$ moves away from $\bar{\bx}$, the partial
derivatives grow but then the behavior of $\phi(t)$ must be such that
their influence is supressed.

We start with the following important claim.\footnote{
We remind the reader that $\nabla F$, the gradient of $F$, is a vector
whose coordinates are the first partial derivatives $\partdiff{F}{x_i}$.
We denote by $\nabla F |_{\bx}$ the gradient evaluated at $\bx$.}

\begin{lemma}
\label{lemma:grad-symmetry}
Assume that $F:[0,1]^X \rightarrow \RR$ is differentiable and invariant under a group of
permutations of coordinates ${\cal G}$. Let $G(\bx) = F(\bar{\bx})$,
where $\bar{\bx} = \E_{\sigma \in {\cal G}}[\sigma(\bx)]$.
Then for any $\bx \in [0,1]^X$,
 $$ \nabla{G}|_\bx = \nabla{F}|_{\bar{\bx}}.$$
\end{lemma}

\begin{proof}
To avoid confusion, we use $\bx$ for the arguments of the functions $F$ and $G$,
and $\bu$, $\bar{\bu}$, etc. for points where their partial derivatives are evaluated.
To rephrase, we want to prove that for any point $\bu$ and any coordinate $i$,
the partial derivatives of $F$ and $G$ evaluated at $\bar{u}$ are equal:
$\partdiff{G}{x_i} \Big|_{\bx=\bu} = \partdiff{F}{x_i} \Big|_{\bx=\bar{\bu}}$.

First, consider $F(\bx)$. We assume that $F(\bx)$ is invariant under
a group of permutations of coordinates $\cal G$,
i.e. $F(\bx) = F(\sigma(\bx))$ for any $\sigma \in {\cal G}$.
Differentiating both sides at $\bx=\bu$, we get by the chain rule:
$$ \partdiff{F}{x_i} \Big|_{\bx=\bu} =
 \sum_j \partdiff{F}{x_j} \Big|_{\bx=\sigma(\bu)} \partdiff{}{x_i} (\sigma(\bx))_j
 = \sum_j \partdiff{F}{x_j} \Big|_{\bx=\sigma(\bu)} \partdiff{x_{\sigma(j)}}{x_i}. $$
Here, $\partdiff{x_{\sigma(j)}}{x_i} = 1$ if $\sigma(j) = i$, and $0$ otherwise.
Therefore,
$$ \partdiff{F}{x_i} \Big|_{\bx=\bu} =
 \partdiff{F}{x_{\sigma^{-1}(i)}} \Big|_{\bx=\sigma(\bu)}. $$
Now, if we evaluate the left-hand side at $\bar{\bu}$, the right-hand side is evaluated
at $\sigma(\bar{\bu}) = \bar{\bu}$, and hence for any $i$ and any $\sigma \in {\cal G}$,
\begin{equation}
\label{eq:F-inv}
\partdiff{F}{x_i} \Big|_{\bx=\bar{\bu}} = \partdiff{F}{x_{\sigma^{-1}(i)}} \Big|_{\bx=\bar{\bu}}.
\end{equation}
Turning to $G(\bx) = F(\bar{\bx})$, let us write $\partdiff{G}{x_i}$
using the chain rule:
$$ \partdiff{G}{x_i} \Big|_{\bx=\bu} = \partdiff{}{x_i} F(\bar{\bx}) \Big|_{\bx=\bu}
 = \sum_j \partdiff{F}{x_j} \Big|_{\bx=\bar{\bu}} \cdot \partdiff{\bar{x}_j}{x_i}. $$
We have $\bar{x}_j = \E_{\sigma \in {\cal G}}[\bx_{\sigma(j)}]$, and
so
$$ \partdiff{G}{x_i} \Big|_{\bx=\bu} = \sum_j \partdiff{F}{x_j} \Big|_{\bx=\bar{\bu}}
 \cdot \partdiff{}{x_i} \E_{\sigma \in {\cal G}}[x_{\sigma(j)}] 
 = \E_{\sigma \in {\cal G}} \left[\sum_j \partdiff{F}{x_j} \Big|_{x=\bar{u}}
 \cdot \partdiff{x_{\sigma(j)}}{x_i} \right]. $$
Again, $\partdiff{x_{\sigma(j)}}{x_i} = 1$ if $\sigma(j) = i$ and $0$ otherwise.
Consequently, we obtain
$$ \partdiff{G}{x_i} \Big|_{\bx=\bu} = \E_{\sigma \in {\cal G}}
 \left[ \partdiff{F}{x_{\sigma^{-1}(i)}} \Big|_{\bx=\bar{\bu}} \right]
 = \partdiff{F}{x_i} \Big|_{\bx=\bar{\bu}} $$
where we used Eq. (\ref{eq:F-inv}) to remove the dependence on $\sigma \in {\cal G}$.
\end{proof}

Observe that the symmetrization operation $\bar{\bx}$ is idempotent, i.e.
$\bar{\bar{\bx}} = \bar{\bx}$.
Because of this, we also get $\nabla{G}|_{\bar{\bx}} = \nabla{F}|_{\bar{\bx}}$.
Note that $G(\bar{\bx}) = F(\bar{\bx})$ follows from the definition,
but it is not obvious that the same holds for gradients, since their
definition involves points where $G(\bx) \neq F(\bx)$. For second partial
derivatives, the equality no longer holds, as can be seen from
a simple example such as $F(x_1,x_2) = 1 - (1-x_1)(1-x_2)$,
$G(x_1,x_2) = 1 - (1-\frac{x_1+x_2}{2})^2$.

Next, we show that the functions $F(\bx)$ and $G(\bx)$ are very similar
in the close vicinity of the region where $\bar{\bx} = \bx$. Recall our definitions:
$H(\bx) = F(\bx) - G(\bx)$, $D(\bx) = ||\bx - \bar{\bx}||^2$.
Based on Lemma~\ref{lemma:grad-symmetry}, we know that $H(\bar{\bx}) = 0$
and $\nabla H|_{\bar{\bx}} = 0$. In the following lemmas,
we present bounds on $H(\bx)$, $D(\bx)$ and their partial derivatives.

\begin{lemma}
\label{lemma:H-bounds}
Let $f:2^X \rightarrow [0,M]$ be invariant under a permutation group $\cal G$.
Let $\bar{\bx} = \E_{\sigma \in {\cal G}}[\sigma(\bx)]$,
$D(\bx) = ||\bx - \bar{\bx}||^2$ and $H(\bx) = F(\bx) - G(\bx)$ where $F(\bx) = \E[f(\hat{\bx})]$
and $G(\bx) = F(\bar{\bx})$. Then
\begin{enumerate}
\item  $ |\mixdiff{H}{x_i}{x_j}| \leq 8M $ everywhere, for all $i,j$;
\item  $ ||\nabla H(\bx)|| \leq 8M|X| \sqrt{D(\bx)}$;
\item  $ |H(\bx)| \leq  8M|X| \cdot D(\bx). $
\end{enumerate}
\end{lemma}

\begin{proof}
First, let us get a bound on the second partial derivatives.
Assuming without loss of generality $x_i=x_j=0$, we have\footnote{$\bx \vee \by$
denotes the coordinate-wise maximum, $(\bx \vee \by)_i = \max \{x_i,y_i\}$ and
$\bx \wedge \by$ denotes the coordinate-wise minimum, $(\bx \wedge \by)_i = \min \{x_i,y_i\}$.}
$$ \mixdiff{F}{x_i}{x_j} =  \E[f(\hat{\bx} \vee (\be_i+\be_j)) -
 f(\hat{\bx} \vee \be_i) - f(\hat{\bx} \vee \be_j) + f(\hat{\bx})] $$
(see \cite{Vondrak08}). Consequently,
$$ \Big| \mixdiff{F}{x_i}{x_j} \Big| \leq 4 \max |f(S)| = 4 M.$$
It is a little bit more involved to analyze $\mixdiff{G}{x_i}{x_j}$.
Since $G(\bx) = F(\bar{\bx})$ and $\bar{\bx} = \E_{\sigma \in {\cal G}}[\sigma(\bx)]$,
we get by the chain rule:
$$ \mixdiff{G}{x_i}{x_j} = \sum_{k,\ell} \mixdiff{F}{x_k}{x_\ell}
 \partdiff{\bar{x}_k}{x_i} \partdiff{\bar{x}_\ell}{x_j} 
 = \E_{\sigma,\tau \in {\cal G}} \left[ \sum_{k,\ell} \mixdiff{F}{x_k}{x_\ell}
   \partdiff{x_{\sigma(k)}}{x_i} \partdiff{x_{\tau(\ell)}}{x_j} \right].$$
It is useful here to use the Kronecker symbol, $\delta_{i,j}$,
which is $1$ if $i=j$ and $0$ otherwise. Note that
$\partdiff{x_{\sigma(k)}}{x_i} = \delta_{i,\sigma(k)} = \delta_{\sigma^{-1}(i),k}$,
etc. Using this notation, we get
$$ \mixdiff{G}{x_i}{x_j} = \E_{\sigma,\tau \in {\cal G}} \left[  \sum_{k,\ell}
\mixdiff{F}{x_k}{x_\ell} \delta_{\sigma^{-1}(i),k} \delta_{\sigma^{-1}(j),\ell} \right] 
 = \E_{\sigma,\tau \in {\cal G}}
 \left[ \mixdiff{F}{x_{\sigma^{-1}(i)}}{x_{\tau^{-1}(j)}} \right], $$
$$ \Big| \mixdiff{G}{x_i}{x_j} \Big| \leq \E_{\sigma,\tau \in {\cal G}}
\left[ \Big| \mixdiff{F}{x_{\sigma^{-1}(i)}}{x_{\tau^{-1}(j)}} \Big| \right]
 \leq 4 M $$
and therefore
$$ \Big| \mixdiff{H}{x_i}{x_j} \Big|
 = \Big| \mixdiff{F}{x_i}{x_j} - \mixdiff{G}{x_i}{x_j} \Big|
 \leq 8 M.$$

Next, we estimate $\partdiff{H}{x_i}$ at a given point $\bu$, depending
on its distance from $\bar{\bu}$. Consider
the line segment between $\bar{\bu}$ and $\bu$. The function $H(\bx) = F(\bx) - G(\bx)$
is $C_\infty$-differentiable, and hence we can apply the mean value theorem
to $\partdiff{H}{x_i}$: There exists a point $\tilde{\bu}$ on the line segment
$[\bar{\bu}, \bu]$ such that
$$ \partdiff{H}{x_i} \Big|_{\bx=\bu} - \partdiff{H}{x_i} \Big|_{\bx=\bar{\bu}}
 = \sum_j \mixdiff{H}{x_j}{x_i} \Big|_{\bx=\tilde{\bu}} (u_j-\bar{u}_j).$$
Recall that $\partdiff{H}{x_i} \Big|_{\bx=\bar{\bu}} = 0$.
Applying the Cauchy-Schwartz inequality to the right-hand side, we get
$$ \left( \partdiff{H}{x_i} \Big|_{\bx=\bu} \right)^2
 \leq \sum_j \left( \mixdiff{H}{x_j}{x_i} \Big|_{\bx=\tilde{\bu}} \right)^2
  || \bu - \bar{\bu} ||^2  \leq (8M)^2 |X| ||\bu - \bar{\bu}||^2.$$
Adding up over all $i \in X$, we obtain
$$ || \nabla H(\bu) ||^2 = \sum_{i}\left( \partdiff{H}{x_i} \Big|_{\bx=\bu} \right)^2
 \leq (8M|X|)^2 ||\bu - \bar{\bu}||^2.$$
Finally, we estimate the growth of $H(\bu)$. Again, by the mean value theorem,
there is a point $\tilde{\bu}$ on the line segment $[\bar{\bu},\bu]$, such that
$$ H(\bu) - H(\bar{\bu}) = (\bu - \bar{\bu}) \cdot \nabla H(\tilde{\bu}).$$
Using $H(\bar{\bu}) = 0$, the Cauchy-Schwartz inequality and the above bound
on $\nabla H$,
$$ (H(\bu))^2 \leq ||\nabla H({\tilde{\bu}})||^2  ||\bu - \bar{\bu}||^2
 \leq (8M|X|)^2 ||\tilde{\bu}-\bar{\bu}||^2 ||\bu-\bar{\bu}||^2. $$
Clearly, $||\tilde{\bu} - \bar{\bu}|| \leq ||\bu - \bar{\bu}||$, and therefore
$$ |H(\bu)| \leq 8M|X| \cdot ||\bu-\bar{\bu}||^2.$$
\end{proof}

\begin{lemma}
\label{lemma:D-bounds}
For the function $D(\bx) = ||\bx - \bar{\bx}||^2$, we have
\begin{enumerate}
\item $\nabla D = 2(\bx - \bar{\bx})$, and therefore
 $||\nabla D|| = 2 \sqrt{D(\bx)}$.
\item For all $i,j$, $|\mixdiff{D}{x_i}{x_j}| \leq 2$.
\end{enumerate}
\end{lemma}

\begin{proof}
Let us write $D(\bx)$ as
$$ D(\bx) = \sum_i (x_i - \bar{x}_i)^2 = \sum_i \E_{\sigma \in {\cal G}}[x_i - x_{\sigma(i)}]
 \E_{\tau \in {\cal G}}[x_i - x_{\tau(i)}]. $$
Taking the first partial derivative,
$$ \partdiff{D}{x_j} = 2 \sum_i \E_{\sigma \in {\cal G}}[x_i - x_{\sigma(i)}]
 \partdiff{}{x_j} \E_{\tau \in {\cal G}}[x_i - x_{\tau(i)}].$$
As before, we have $\partdiff{x_i}{x_j} = \delta_{ij}$.
Using this notation, we get
\begin{eqnarray*}
\partdiff{D}{x_j} & = & 2 \sum_i \E_{\sigma \in {\cal G}}[x_i - x_{\sigma(i)}]
 \E_{\tau \in {\cal G}}[\delta_{ij} - \delta_{\tau(i),j}] \\
 & = & 2 \sum_i \E_{\sigma,\tau \in {\cal G}}[(x_i - x_{\sigma(i)})
    (\delta_{ij} - \delta_{i,\tau^{-1}(j)})] \\
 & = & 2 \ \E_{\sigma,\tau \in {\cal G}}[x_j - x_{\sigma(j)} - x_{\tau^{-1}(j)}
    + x_{\sigma(\tau^{-1}(j))}].
\end{eqnarray*}
Since the distributions of $\sigma(j)$, $\tau^{-1}(j)$ and $\sigma(\tau^{-1}(j))$
are the same, we obtain
$$ \partdiff{D}{x_j} = 2 \ \E_{\sigma \in {\cal G}}[x_j - x_{\sigma(j)}]
 = 2(x_j - \bar{x}_j) $$
and
$$ ||\nabla D||^2 = \sum_j \Big| \partdiff{D}{x_j} \Big|^2
 = 4 \sum_j (x_j - \bar{x}_j)^2 = 4 D(\bx).$$

Finally, the second partial derivatives are
$$ \mixdiff{D}{x_i}{x_j} = 2 \partdiff{}{x_i} (x_j - \bar{x}_j)
 = 2 \partdiff{}{x_i} \E_{\sigma \in {\cal G}}[x_j - x_{\sigma(j)}]
 = 2 \E_{\sigma \in {\cal G}}[\delta_{ij} - \delta_{i,\sigma(j)}] $$
which is clearly bounded by $2$ in the absolute value.
\end{proof}

Now we come back to $\tilde{F}(\bx)$ and its partial derivatives.
Recall equations (\ref{eq:partdiff}) and (\ref{eq:mixdiff}).
The problematic terms are those involving $\phi'(D(\bx))$ and $\phi''(D(\bx))$.
Using our bounds on $H(\bx)$, $D(\bx)$ and their derivatives, however,
we notice that $\phi'(D(\bx))$ always appears with factors on the order
of $D(\bx)$ and $\phi''(D(\bx))$ appears with factors on the order of $(D(\bx))^2$.
Thus, it is sufficient if $\phi(t)$ is defined so that we have control
over $t \phi'(t)$ and $t^2 \phi''(t)$. The following lemma describes
the function that we need.

\begin{lemma}
\label{lemma:phi-construction}
For any $\alpha, \beta > 0$, there is $\delta \in (0, \beta)$
and a function $\phi:\RR_+ \rightarrow [0,1]$ with an absolutely continuous first derivative
such that
\begin{enumerate}
\item For $t \leq \delta$, $\phi(t) = 1$.
\item For $t \geq \beta$, $\phi(t) < e^{-1/\alpha}$.
\item For all $t \geq 0$, $|t \phi'(t)| \leq 4 \alpha$.
\item For almost all $t \geq 0$, $|t^2 \phi''(t)| \leq 10 \alpha$.
\end{enumerate}
\end{lemma}

\begin{proof}
First, observe that if we prove the lemma for some particular value $\beta>0$, we can also prove prove it
for any other value $\tilde{\beta}>0$, by modifying the function as follows: $\tilde{\phi}(t) = \phi(\beta t / \beta')$.
This corresponds to a scaling of the parameter $t$ by $\beta' / \beta$. Observe that then $|t \tilde{\phi}'(t)| = |\frac{\beta}{\beta'} t \phi'(\beta t / \beta')| \leq 4 \alpha$ and $|t^2 \tilde{\phi}''(t)| = |(\frac{\beta}{\beta'})^2 t^2 \phi''(\beta t / \beta')| \leq 10 \alpha$, so the conditions are still satisfied.

Therefore, we can assume without
loss of generality that $\beta>0$ is a value of our choice,
for example $\beta = e^{1/(2\alpha^2)}+1$.
If we want to prove the result for a different value of $\beta$,
we can just scale the argument $t$ and the constant $\delta$
by $\beta / (e^{1/(2\alpha^2)}+1)$; the bounds on
$t \phi'(t)$ and $t^2 \phi''(t)$ still hold.

We can assume that $\alpha \in (0,\frac18)$ because for larger $\alpha$,
the statement only gets weaker. As we argued, we can assume WLOG that
$\beta = e^{1/(2\alpha^2)}+1$. We set $\delta = 1$ and
$\delta_2 = 1 + (1+\alpha)^{-1/2} \leq 2$.
(We remind the reader that in general these values will be scaled depending
on the actual value of $\beta$.) We define the function as follows:
\begin{enumerate}
\item $\phi(t) = 1$ for $t \in [0,\delta]$.
\item $\phi(t) = 1 - \alpha (t-1)^2$ for $t \in [\delta, \delta_2]$.
\item $\phi(t) = (1+\alpha)^{-1-\alpha} (t - 1)^{-2 \alpha}$
 for $t \in [\delta_2, \infty)$.
\end{enumerate}

\noindent
Let's verify the properties of $\phi(t)$.
For $t \in [0,\delta]$, we have
$\phi'(t) = \phi''(t) = 0$. For $t \in [\delta,\delta_2]$, we have
$$ \phi'(t) = -2\alpha \left(t-1 \right), \ \ \ \ \ \ 
 \phi''(t) = -2\alpha, $$
and for $t \in [\delta_2, \infty)$,
$$ \phi'(t) = -{2\alpha}{(1+\alpha)^{-1-\alpha}} \left( t-1 \right)^{-2\alpha-1}, $$
$$ \phi''(t) = {2\alpha(1+2\alpha)}{(1+\alpha)^{-1-\alpha}} \left( t-1 \right)^{-2\alpha-2}. $$
First, we check that the values and first derivatives agree at the breakpoints.
For $t = \delta = 1$, we get $\phi(1) = 1$ and $\phi'(1) = 0$.
For $t = \delta_2 = 1 + (1+\alpha)^{-1/2}$, we get
$\phi(\delta_2) = (1+\alpha)^{-1}$ and
$\phi'(\delta_2) = -2 \alpha (1+\alpha)^{-1/2}$.
Next, we need to check is that $\phi(t)$ is very small
for $t \geq \beta$. The function is decreasing for $t > \beta$,
therefore it is enough to check $t = \beta = e^{1/(2\alpha^2)}+1$:
$$ \phi\left(\beta \right) = (1+\alpha)^{-1-\alpha} (\beta - 1)^{-2\alpha}
  \leq (\beta-1)^{-2\alpha} = e^{-1/\alpha}. $$
The derivative bounds are satisfied trivially for $t \in [0,\delta]$.
For $t \in [\delta, \delta_2]$, using $t \leq \delta_2 = 1 + (1+\alpha)^{-1/2}$,
$$ |t \phi'(t)| = t \cdot 2 \alpha (t-1)
 \leq 2 \alpha (1 + (1+\alpha)^{-1/2}) (1+\alpha)^{-1/2} \leq 4 \alpha $$
and using $\alpha \in (0,\frac18)$,
$$ |t^2 \phi''(t)| = t^2 \cdot 2 \alpha \leq 2 \alpha (1 + (1+\alpha)^{-1/2})^2
 \leq 8 \alpha.$$
For $t \in [\delta_2,\infty)$, using $t-1 \geq (1+\alpha)^{-1/2}$,
\begin{eqnarray*}
|t \phi'(t)| & = & t \cdot \frac{2 \alpha}{(1+\alpha)^{1+\alpha}} \left(t-1 \right)^{-2\alpha-1} 
 = \frac{2 \alpha}{1+\alpha} \left( (1 + \alpha)
\left( t-1 \right)^2 \right)^{-\alpha} \frac{t}{t-1} \\
& \leq & \frac{2 \alpha}{1 + \alpha} \cdot \frac{t}{t-1}
 \leq \frac{2 \alpha}{1+\alpha} \cdot \frac{1 + (1+\alpha)^{-1/2}}{(1+\alpha)^{-1/2}}
 = 2 \alpha \cdot \frac{1 + (1+\alpha)^{-1/2}}{(1+\alpha)^{1/2}} \leq 4 \alpha
\end{eqnarray*}
and finally, using $\alpha \in (0,\frac18)$,
\begin{eqnarray*}
 |t^2 \phi''(t)| & = & t^2 \cdot \frac{2\alpha(1+2\alpha)}{(1+\alpha)^{1+\alpha}}
 \left( t-1 \right)^{-2\alpha-2} 
 = \frac{2\alpha(1+2\alpha)}{1+\alpha}
 \left( (1 + \alpha) \left( t-1 \right)^2 \right)^{-\alpha}
 \left( \frac{t}{t-1} \right)^2 \\
& \leq & \frac{2\alpha(1+2\alpha)}{1+\alpha} \cdot
 \left( \frac{1 + (1+\alpha)^{-1/2}}{(1+\alpha)^{-1/2}} \right)^2
 \leq 8 \alpha (1 + 2 \alpha) \leq 10 \alpha.
\end{eqnarray*}
\end{proof}

Using Lemmas~\ref{lemma:H-bounds}, \ref{lemma:D-bounds} and \ref{lemma:phi-construction},
we now prove bounds on the derivatives of $\tilde{F}(\bx)$.

\begin{lemma}
\label{lemma:F-bounds}
Let $ \tilde{F}(\bx) = (1-\phi(D(\bx))) F(\bx) + \phi(D(\bx)) G(\bx) $
where $F(\bx) = \E[f(\hat{\bx})]$, $f:2^X \rightarrow [0,M]$,
$G(\bx) = F(\bar{\bx})$, $D(\bx)=||\bx-\bar{\bx}||^2$
are as above and $\phi(t)$ is provided by Lemma~\ref{lemma:phi-construction}
for a given $\alpha>0$.
Then, whenever $\mixdiff{F}{x_i}{x_j} \leq 0$,
$$ 
 \mixdiff{\tilde{F}}{x_i}{x_j}
  \leq  512 M |X| \alpha. $$
If, in addition, $\partdiff{F}{x_i} \geq 0$, then
$$
\partdiff{\tilde{F}}{x_i}
  \geq  -64 M |X| \alpha. $$
\end{lemma}

\begin{proof}
We have
$ \tilde{F}(\bx) = F(\bx) - \phi(D(\bx)) H(\bx) $
where $H(\bx) = F(\bx) - G(\bx)$. By differentiating once, we get
$$ \partdiff{\tilde{F}}{x_i} = \partdiff{F}{x_i}
 - \phi(D(\bx)) \partdiff{H}{x_i} - \phi'(D(\bx)) \partdiff{D}{x_i} H(\bx), $$
i.e.
$$ \Big| \partdiff{\tilde{F}}{x_i} -
 \left(\partdiff{F}{x_i} - \phi(D(\bx)) \partdiff{H}{x_i} \right) \Big|
 = \Big| \phi'(D(\bx)) \partdiff{D}{x_i} H(\bx) \Big|.$$
By Lemma~\ref{lemma:H-bounds} and \ref{lemma:D-bounds},
we have $|\partdiff{D}{x_i}| = 2 |x_i - \bar{x}_i| \leq 2$
and $|H(\bx)| \leq 8M|X| \cdot D(\bx)$. Therefore,
$$ \Big| \partdiff{\tilde{F}}{x_i} -
 \left(\partdiff{F}{x_i} - \phi(D(\bx)) \partdiff{H}{x_i} \right) \Big|
 \leq 16 M |X| D(\bx) \cdot \Big| \phi'(D(\bx)) \Big|.$$
By Lemma~\ref{lemma:phi-construction}, $|D(\bx) \phi'(D(\bx))| \leq 4 \alpha$, and hence
$$ \Big| \partdiff{\tilde{F}}{x_i} -
 \left(\partdiff{F}{x_i} - \phi(D(\bx)) \partdiff{H}{x_i} \right) \Big|
 \leq 64 M |X| \alpha.$$
Assuming that $\partdiff{F}{x_i} \geq 0$, we also have $\partdiff{G}{x_i} \geq 0$
(see Lemma~\ref{lemma:grad-symmetry}) and therefore,
$ \partdiff{F}{x_i} - \phi(D(\bx)) \partdiff{H}{x_i}
 = (1-\phi(D(\bx))) \partdiff{F}{x_i} + \phi(D(\bx)) \partdiff{G}{x_i} \geq 0$.
Consequently,
$$ \partdiff{\tilde{F}}{x_i} \geq - 64 M |X| \alpha.$$
By differentiating $\tilde{F}$ twice, we obtain
\begin{eqnarray*}
\mixdiff{\tilde{F}}{x_i}{x_j} & = &
 \mixdiff{F}{x_i}{x_j} - \phi(\DD(\bx)) \mixdiff{H}{x_i}{x_j}
 - \phi''(\DD(\bx)) \partdiff{\DD}{x_i} \partdiff{\DD}{x_j} H(\bx) \\
& &  - \phi'(\DD(\bx)) \left( \partdiff{\DD}{x_j} \partdiff{H}{x_i}
 + \mixdiff{\DD}{x_i}{x_j} H(\bx)
 + \partdiff{\DD}{x_i} \partdiff{H}{x_j} \right).
\end{eqnarray*}
Again, we use Lemma~\ref{lemma:H-bounds} and \ref{lemma:D-bounds}
to bound $|H(\bx)| \leq 8M|X| D(\bx)$, $|\partdiff{H}{x_i}| \leq 8M|X| \sqrt{D(\bx)}$,
 $|\mixdiff{H}{x_i}{x_j}| \leq 8M$, $|\partdiff{D}{x_i}| \leq 2 \sqrt{D(\bx)}$
and $|\mixdiff{D}{x_i}{x_j}| \leq 2$. We get
\begin{eqnarray*}
\Big| \mixdiff{\tilde{F}}{x_i}{x_j} - \mixdiff{F}{x_i}{x_j} +
 \phi(\DD(\bx)) \mixdiff{H}{x_i}{x_j} \Big|
 \leq 32 M|X| \Big| D^2(\bx) \phi''(\DD(\bx)) \Big|
  + 48 M|X| \Big| D(\bx) \phi'(\DD(\bx)) \Big|
\end{eqnarray*}
Observe that $\phi'(D(\bx))$ appears with $D(\bx)$ and
$\phi''(D(\bx))$ appears with $(D(\bx))^2$.
By Lemma~\ref{lemma:phi-construction}, $|D(\bx) \phi'(D(\bx))| \leq 4 \alpha$
and $|D^2(\bx) \phi''(D(\bx))| \leq 10 \alpha$. Therefore,
\begin{eqnarray*}
\Big| \mixdiff{\tilde{F}}{x_i}{x_j} - \left( \mixdiff{F}{x_i}{x_j}
 - \phi(\DD(\bx)) \mixdiff{H}{x_i}{x_j} \right) \Big|
 \leq  320 M|X|\alpha + 192 M|X| \alpha  = 512 M |X| \alpha.
\end{eqnarray*}
If $\mixdiff{F}{x_i}{x_j} \leq 0$ for all $i,j$, then also $\mixdiff{G}{x_i}{x_j} =
 \E_{\sigma,\tau \in \cG}[\mixdiff{F}{x_{\sigma^{-1}(i)}}{x_{\tau^{-1}(j)}}] \leq 0$
 (see the proof of Lemma~\ref{lemma:H-bounds}). Also, 
$\mixdiff{F}{x_i}{x_j} - \phi(\DD(\bx)) \mixdiff{H}{x_i}{x_j}
 = \phi(D(\bx)) \mixdiff{F}{x_i}{x_j} + (1-\phi(D(\bx))) \mixdiff{G}{x_i}{x_j} \leq 0$.
We obtain
\begin{eqnarray*}
\mixdiff{\tilde{F}}{x_i}{x_j} \leq 512 M |X| \alpha.
\end{eqnarray*}
\end{proof}

Finally, we can finish the proof of Lemma~\ref{lemma:final-fix}.

\begin{proof}[Proof of Lemma~\ref{lemma:final-fix}]
Let $\epsilon>0$ and $f:2^X \rightarrow [0,M]$.
We choose $\beta = \frac{\epsilon}{16 M|X|}$, so that $|H(\bx)| = |F(\bx) - G(\bx)| \leq \epsilon/2$
whenever $D(\bx) = ||\bx - \bar{\bx}||^2 \leq \beta$
(due to by Lemma~\ref{lemma:H-bounds}, which states that $|H(\bx)| \leq 8 M |X| D(\bx)$).
Also, let $\alpha = \frac{\epsilon}{2000 M |X|^3}$.
For these values of $\alpha,\beta>0$,
let $\delta > 0$ and $\phi:\RR_+ \rightarrow [0,1]$
be provided by Lemma~\ref{lemma:phi-construction}.
We define
$$ \tilde{F}(\bx) = (1-\phi(D(\bx))) F(\bx) + \phi(D(\bx)) G(\bx). $$
Lemma~\ref{lemma:F-bounds} provides bounds on the first and second
partial derivatives of $\tilde{F}(\bx)$.
Finally, we modify $\tilde{F}(\bx)$ so that it satisfies the required
conditions (submodularity and optionally monotonicity).
For that purpose, we add a suitable multiple of the following function:
$$ J(\bx) = |X|^2 + 3|X| \sum_{i \in X} x_i - \left(\sum_{i \in X} x_i \right)^2.$$
We have $0 \leq J(\bx) \leq 3|X|^2$,
$\partdiff{J}{x_i} = 3|X| - 2 \sum_{i \in X} x_i \geq |X|$.
Further, $\mixdiff{J}{x_i}{x_j} = -2$. Note also that $J(\bar{\bx}) = J(\bx)$,
since $J(\bx)$ depends only on the sum of all coordinates $\sum_{i \in X} x_i$.
To make $\tilde{F}(\bx)$ submodular and optionally monotone, we define:
$$ \hat{F}(\bx) = \tilde{F}(\bx) + 256 M |X| \alpha J(\bx),$$
$$ \hat{G}(\bx) = G(\bx) + 256 M |X| \alpha J(\bx). $$
We verify the properties of $\hat{F}(\bx)$ and $\hat{G}(\bx)$:
\begin{enumerate}
\item 
For any $\bx \in P({\cal F})$, we have
\begin{eqnarray*}
\hat{G}(\bx) & = & G(\bx) + 256 M |X| \alpha J(\bx) \\
 & = & F(\bar{\bx}) + 256 M |X| \alpha J(\bar{\bx}) \\
 & = & \hat{F}(\bar{\bx}).
\end{eqnarray*}

\item When $D(\bx) = ||\bx - \bar{\bx}||^2 \geq \beta$,
Lemma~\ref{lemma:phi-construction} guarantees
that $0 \leq \phi(D(\bx)) < e^{-1/\alpha} \leq \alpha$ and
\begin{eqnarray*}
|\hat{F}(\bx) - F(\bx)| & \leq & \phi(D(\bx)) |G(\bx) - F(\bx)| + 256 M |X| \alpha J(\bx) \\
 & \leq & \alpha M + 768 M |X|^3 \alpha \\
 & \leq & \epsilon
\end{eqnarray*}
using $0 \leq F(\bx), G(\bx) \leq M$, $|X| \leq J(\bx) \leq 3|X|^2$ and $\alpha = \frac{\epsilon}{2000 M |X|^3}$.

When $D(\bx) = ||\bx - \bar{\bx}||^2 < \beta$, we chose the value of $\beta$ so that
$|G(\bx) - F(\bx)| < \epsilon/2$ and so by the above,
\begin{eqnarray*}
|\hat{F}(\bx) - F(\bx)| & \leq & \phi(D(\bx)) |G(\bx) - F(\bx)| + 256 M |X| \alpha J(\bx) \\
 & \leq & \epsilon/2 + 768 M |X|^3 \alpha \\
 & \leq & \epsilon.
\end{eqnarray*}

\item Due to Lemma~\ref{lemma:phi-construction}, $\phi(t) = 1$ for $t \in [0,\delta]$.
Hence, whenever $D(\bx) = ||\bx - \bar{\bx}||^2 \leq \delta$, we have
$ \tilde{F}(\bx) = G(\bx) = F(\bar{\bx})$, which depends only on $\bar{\bx}$.
Also, we have $\hat{F}(\bx) = \hat{G}(\bx) = F(\bar{\bx}) + 256 M |X| \alpha J(\bx)$
and again, $J(\bx)$ depends only on $\bar{\bx}$ (in fact, only on the average
of all coordinates of $\bx$). Therefore, $\hat{F}(\bx)$ and $\hat{G}(\bx)$
in this case depend only on $\bar{\bx}$.

\item The first partial derivatives of $\hat{F}$ are given by the formula
$$ \partdiff{\hat{F}}{x_i} = \partdiff{F}{x_i}
 - \phi(D(\bx)) \partdiff{H}{x_i} - \phi'(D(\bx)) \partdiff{D}{x_i} H(\bx) + 256 M|X|\alpha \partdiff{J}{x_i}. $$
The functions $F, H, D, J$ are infinitely differentiable, so the only possible issue is with $\phi$.
By inspecting our construction of $\phi$ (Lemma~\ref{lemma:phi-construction}), we can see that it is piecewise 
infinitely differentiable, and $\phi'$ is continuous at the breakpoints. Therefore, it is also
absolutely continuous. This implies that $\partdiff{\hat{F}}{x_i}$ is absolutely continuous.

The function $\hat{G}(\bx) = F(\bar{\bx}) + 256 M|X|\alpha J(\bx)$ is infinitely differentiable,
so its first partial derivatives are also absolutely continuous.

\item Assuming $\partdiff{F}{x_i} \geq 0$, we get $\partdiff{\tilde{F}}{x_i}
 \geq - 64 M |X| \alpha$ by Lemma~\ref{lemma:F-bounds}. Using $\partdiff{J}{x_i}
 \geq |X|$, we get $\partdiff{\hat{F}}{x_i} = \partdiff{\tilde{F}}{x_i}
 + 256 M|X| \alpha J(\bx) \geq 0$. The same holds for $\partdiff{\hat{G}}{x_i}$
since $\partdiff{G}{x_i} \geq 0$.

\item Assuming $\mixdiff{F}{x_i}{x_j} \leq 0$, we get $\mixdiff{\tilde{F}}{x_i}{x_j}
 \leq 512 M |X| \alpha$ by  Lemma~\ref{lemma:F-bounds}. Using $\mixdiff{J}{x_i}{x_j}
 = -2$, we get $\mixdiff{\hat{F}}{x_i}{x_j} = \mixdiff{\tilde{F}}{x_i}{x_j}
 + 256 M |X| \alpha \mixdiff{J}{x_i}{x_j} \leq 0$. The same holds for
$\mixdiff{\hat{G}}{x_i}{x_j}$ since $\mixdiff{G}{x_i}{x_j} \leq 0$.
\end{enumerate}
\end{proof}

This concludes the proofs of our main hardness results.

\section{Algorithms using the multilinear relaxation}
\label{section:algorithms}

Here we turn to our algorithmic results. First, we discuss
the problem of maximizing a submodular (but not necessarily
monotone) function subject to a matroid independence constraint.

\subsection{Matroid independence constraint}
\label{section:submod-independent}

Consider the problem $\max \{ f(S): S \in \cI \}$, where $\cI$ is the collection
of independent sets in a matroid $\cM$.
We design an algorithm based on the multilinear relaxation of the problem,
$\max \{ F(\bx): \bx \in P(\cM) \}$.
Our algorithm can be seen as "continuous local search" in the matroid polytope $P(\cM)$,
constrained in addition by the box $[0,t]^X$
for some fixed $t \in [0,1]$. The intuition is that this forces
our local search to use fractional solutions that are more fuzzy
than integral solutions and therefore less likely to get stuck in a local
optimum. On the other hand, restraining the search space too much would not
give us much freedom in searching for a good fractional point.
This leads to a tradeoff and an optimal choice of $t \in [0,1]$
which we leave for later.

The matroid polytope is defined as
$ P(\cM) = \mbox{conv} \{ \b1_I: I \in \cI \} $.
or equivalently \cite{E70} as
$ P(\cM) = \{ \bx \geq 0: \forall S; \sum_{i \in S} x_i \leq r_\cM(S) \}$,
where $r_\cM(S)$ is the rank function of $\cM$.
We define
$$ P_t(\cM) = P(\cM) \cap [0,t]^X = \{ \bx \in P(\cM): \forall i; x_i \leq t \}.$$
We consider the problem $\max \{F(\bx): \bx \in P_t(\cM)\}$.
We remind the reader that
$F(\bx) = \E[f(\hat{\bx})]$ denotes the multilinear extension.
Our algorithm works as follows.

\

\noindent{\bf Fractional local search in $P_t(\cM)$} \\
 (given $t = \frac{r}{q}$, $r \leq q$ integer)
\begin{enumerate}
\item Start with $\bx := (0,0,\ldots,0)$. Fix $\delta = 1/q$.
\item If there is  $i,j \in X$ and a direction ${\bf v} \in
 \{\be_j, -\be_i, \be_j - \be_i \}$ such that
$\bx + \delta {\bf v} \in P_t(\cM)$ and $F(\bx + \delta {\bf v}) > F(\bx)$,
set $\bx := \bx + \delta {\bf v}$ and repeat.
\item If there is no such direction ${\bf v}$, apply pipage rounding to $\bx$
and return the resulting solution.
\end{enumerate}

\noindent{\em Notes.}
The procedure as presented here would not run in polynomial
time. A modification which runs in polynomial time is that
we move to a new solution only if $F(\bx+\delta {\bf v}) >
 F(\bx) + \frac{\delta}{poly(n)} OPT$ (where we first get
a rough estimate of $OPT$ using previous methods).
For simplicity, we analyze the variant above
and finally discuss why we can modify it without
losing too much in the approximation factor. We also defer
the question of how to estimate the value of $F(\bx)$
to the end of this section.

For $t=1$, we have $\delta=1$ and the procedure
reduces to discrete local search. However, it is known that discrete
local search alone does not give any approximation guarantee. With
additional modifications, an algorithm based on discrete local search
achieves a $(\frac14-o(1))$-approximation \cite{LMNS09}.

Our version of fractional local search avoids this issue and leads
directly to a good fractional solution.  Throughout the algorithm,
we maintain $\bx$ as a linear combination of $q$ independent sets
such that no element appears in more than $r$ of them. A local step
corresponds to an add/remove/switch operation preserving this condition.

Finally, we use pipage rounding to convert a fractional solution into
an integral one. As we show in Lemma~\ref{lemma:pipage} (in the Appendix),
a modification of the technique from \cite{CCPV07}
can be used to find an integral solution without any loss in the objective
function.

\begin{theorem}
\label{thm:submod-matroid-approx}
The fractional local search algorithm for any fixed $t \in [0,\frac12 (3-\sqrt{5})]$
returns a solution of value at least $(t - \frac12 t^2) OPT$,
where $OPT = \max \{f(S): S \in \cI\}$.
\end{theorem}

We remark that for $t=\frac12(3-\sqrt{5})$, we would obtain a $\frac14(-1+\sqrt{5})
 \simeq 0.309$-approximation, improving the factor of $\frac14$ \cite{LMNS09}.
This is not a rational value, but we can pick a rational $t$
arbitrarily close to $\frac12 (3-\sqrt{5})$.
For values $t > \frac12 (3-\sqrt{5})$, our analysis does not yield
a better approximation factor.

First, we discuss properties of the point found by the fractional local search algorithm.

\begin{lemma}
\label{lemma:local-opt}
The outcome of the fractional local search algorithm $x$ is
a ``fractional local optimum'' in the following sense.
(All the partial derivatives are evaluated at $x$.)
\begin{itemize}
\item For any $i$ such that $\bx - \delta \be_i \in P_t(\cM)$,
$\partdiff{F}{x_i} \geq 0.$
\item For any $j$ such that $\bx + \delta \be_j \in P_t(\cM)$,
$\partdiff{F}{x_j} \leq 0.$
\item For any $i,j$ such that $\bx + \delta (\be_j - \be_i) \in P_t(\cM)$,
$\partdiff{F}{x_j} - \partdiff{F}{x_i} \leq 0.$
\end{itemize}
\end{lemma}

\begin{proof}
We use the property (see \cite{CCPV07}) that along any direction $\bv = \pm \be_i$
or $\bv = \be_i - \be_j$, the function $F(\bx + \lambda \bv)$ is a convex
function of $\lambda$. Also, observe that if it is possible to move from $\bx$
in the direction of $\bv$ by any nonzero amount, then it is possible to move
by $\delta \bv$, because all coordinates of $\bx$ are integer multiples of $\delta$
and all the constraints also have coefficients which are integer multiples of $\delta$.
Therefore, if $\frac{dF}{d\lambda} > 0$ and it is possible
to move in the direction of $\bv$, we would get
$F(\bx + \delta \bv) > F(\bx)$ and the fractional local search would continue.

If $\bv = -\be_i$ and it is possible to move along $-\be_i$,
we get $\frac{dF}{d\lambda} = -\partdiff{F}{x_i} \leq 0$. Similarly,
if $\bv = \be_j$ and it is possible to move along $\be_j$,
we get $\frac{dF}{d\lambda} = \partdiff{F}{x_j} \leq 0$. Finally,
if $\bv = \be_j - \be_i$ and it is possible to move along $\be_j - \be_i$,
we get $\frac{dF}{d\lambda} = \partdiff{F}{x_j} - \partdiff{F}{x_i} \leq 0$.
\end{proof}

We refer to the following exchange property for matroids
(which follows easily from \cite{Schrijver}, Corollary 39.12a; see also \cite{LMNS09}).

\begin{lemma}
\label{lemma:exchange}
If $I, C \in \cI$, then for any $j \in C \setminus I$,
there is $\pi(j) \subseteq I \setminus C$, $|\pi(j)| \leq 1$,
such that $I \setminus \pi(j) + j \in \cI$.
Moreover, the sets $\pi(j)$ are disjoint
(each $i \in I \setminus C$ appears at most once as $\pi(j) = \{i\}$).
\end{lemma}

Using this, we prove a lemma about fractional local optima
which generalizes Lemma 2.2 in \cite{LMNS09}.

\begin{lemma}
\label{lemma:fractional-search}
Let $\bx$ be the outcome of fractional local search over $P_t(\cM)$.
Let $C \in \cI$ be any independent set.
Let $C' = \{i \in C: x_i < t\}$. Then
$$ 2 F(\bx) \geq F(\bx \vee \b1_{C'}) + F(\bx \wedge \b1_C).$$
\end{lemma}

Note that for $t = 1$, the lemma reduces to
$2 F(\bx) \geq F(\bx \vee \b1_C) + F(\bx \wedge \b1_C)$
(similar to Lemma 2.2 in \cite{LMNS09}).
For $t < 1$, however, it is necessary to replace $C$ by $C'$
in the first expression, which becomes apparent in the proof.
The reason is that we do not have any information on $\partdiff{F}{x_i}$
for coordinates where $x_i = t$.

\begin{proof}
Let $C \in \cI$ and assume $\bx \in P_t(\cM)$ is a local optimum.
Since $\bx \in P(\cM)$, we can decompose it into a convex linear
combination of vertices of $P(\cM)$, $\bx = \sum_{I \in \cI} x_I \b1_I$
where $\sum x_I = 1$..
By the smooth submodularity of $F(\bx)$ (see \cite{Vondrak08}),
\begin{eqnarray*}
F(\bx \vee \b1_{C'}) - F(\bx) \leq \sum_{j \in C'} (1 - x_j) \partdiff{F}{x_j} 
  =  \sum_{j \in C'} \sum_{I: j \notin I} x_I \partdiff{F}{x_j}
  = \sum_I x_I \sum_{j \in C' \setminus I} \partdiff{F}{x_j}.
\end{eqnarray*}
All partial derivatives here are evaluated at $\bx$.
On the other hand, also by submodularity,
\begin{eqnarray*}
 F(\bx) - F(\bx \wedge \b1_C) \geq \sum_{i \notin C} x_i \partdiff{F}{x_i}
 = \sum_{i \notin C} \sum_{I: i \in I} x_I \partdiff{F}{x_i}
 = \sum_I x_I \sum_{i \in I \setminus C} \partdiff{F}{x_i}.
\end{eqnarray*}
To prove the lemma, it remains to prove the following.

\

\noindent {\bf Claim.} Whenever $x_I > 0$,
$\sum_{j \in C' \setminus I} \partdiff{F}{x_j}
 \leq \sum_{i \in I \setminus C} \partdiff{F}{x_i}$.

\

\noindent {\em Proof:}
For any $I \in \cI$, we can apply
Lemma~\ref{lemma:exchange} to get a mapping $\pi$ such that
$I \setminus \pi(j) + j \in \cI$ for any $j \in C \setminus I$.
Now, consider $j \in C' \setminus I$, i.e. $j \in C \setminus I$ and $x_j < t$.

If $\pi(j) = \emptyset$, is possible to move from $\bx$ in the direction
of $\be_j$, because $I + j \in \cI$ and hence we can replace $I$ by $I+j$
(or at least we can do this for some nonzero fraction of its coefficient)
in the linear combination.
Because $x_j < t$, we can move by a nonzero amount inside $P_t(\cM)$.
By Lemma~\ref{lemma:local-opt}, $\partdiff{F}{x_j} \leq 0$.

Similarly, if $\pi(j) = \{i\}$, it is possible to move in the direction
of $\be_j - \be_i$, because $I$ can be replaced by $I \setminus \pi(j) + i$
for some nonzero fraction of its coefficient. By Lemma~\ref{lemma:local-opt},
in this case $\partdiff{F}{x_j} - \partdiff{F}{x_i} \leq 0$.

Finally, for any $i \in I$ we have $x_i > 0$ and therefore we can decrease
$x_i$ while staying inside $P_t(\cM)$. By Lemma~\ref{lemma:local-opt},
we have $\partdiff{F}{x_i} \geq 0$ for all $i \in I$.
This means
$$ \sum_{j \in C' \setminus I} \partdiff{F}{x_j}
 \leq  \sum_{j \in C' \setminus I: \pi(j) = \emptyset} \partdiff{F}{x_j} +
  \sum_{j \in C' \setminus I: \pi(j) = \{i\}} \partdiff{F}{x_i}
 \leq \sum_{i \in I \setminus C} \partdiff{F}{x_i} $$
using the inequalities we derived above, and the fact that
each $i \in I \setminus C$ appears at most once in $\pi(j)$.
This proves the Claim, and hence the Lemma.
\end{proof}

Now we are ready to prove Theorem~\ref{thm:submod-matroid-approx}.

\begin{proof}
Let $\bx$ be the outcome of the fractional local search over $P_t(\cM)$.
Define $A = \{ i: x_i = t \}$.
Let $C$ be the optimum solution and $C' = C \setminus A = \{ i \in C: x_i < t\}$.
By Lemma~\ref{lemma:fractional-search},
$$ 2 F(\bx) \geq F(\bx \vee \b1_{C'}) + F(\bx \wedge \b1_C).$$
First, let's analyze $F(\bx \wedge \b1_C)$. We apply Lemma~\ref{lemma:rnd-threshold} (in the Appendix),
which states that $F(\bx \wedge \b1_C) \geq \E[f(T(\bx \wedge \b1_C))]$. Here,
$T(\bx \wedge \b1_C)$ is a random threshold set corresponding to the
vector $\bx \wedge \b1_C$, i.e.
$$ T(\bx \wedge \b1_C) = \{ i: (\bx \wedge \b1_C)_i > \lambda \}
 = \{ i \in C: x_i > \lambda \} = T(\bx) \cap C $$
where $\lambda \in [0,1]$ is uniformly random.
Equivalently,
$$ F(\bx \wedge \b1_C) \geq \E[f(T(\bx) \cap C)].$$
Due to the definition of a threshold set, with probability $t$ we have
$\lambda < t$ and $T(\bx)$ contains $A = \{ i: x_i = t \} = C \setminus C'$.
Then, $f(T(\bx) \cap C) + f(C') \geq f(C)$ by submodularity.
We conclude that
\begin{equation}
\label{eq:1}
F(\bx \wedge \b1_C) \geq t (f(C) - f(C')).
\end{equation}
Next, let's analyze $F(\bx \vee \b1_{C'})$. We consider the ground set
partitioned into $X = C \cup \bar{C}$, and we apply Lemma~\ref{lemma:submod-split} (in the Appendix).
(Here, $\bar{C}$ denotes $X \setminus C$, the complement of $C$ inside $X$.)
We get
$$ F(\bx \vee \b1_{C'}) \geq \E[f((T_1(\bx \vee \b1_{C'}) \cap C)
 \cup (T_2(\bx \vee \b1_{C'}) \cap \bar{C}))].$$
The random threshold sets look as follows: $T_1(\bx \vee \b1_{C'}) \cap C
 = (T_1(\bx) \cup C') \cap C$ is equal to $C$ with probability $t$,
and equal to $C'$ otherwise. $T_2(\bx \vee \b1_{C'}) \cap \bar{C}
 = T_2(\bx) \cap \bar{C}$ is empty with probability $1-t$. 
(We ignore the contribution when $T_2(\bx) \cap \bar{C} \neq \emptyset$.)
Because $T_1$ and $T_2$ are independently sampled, we get
$$ F(\bx \vee \b1_{C'}) \geq t(1-t) f(C) + (1-t)^2 f(C').$$
Provided that $t \in [0,\frac12 (3 - \sqrt{5})]$, we have $t \leq (1-t)^2$.
Then, we can write
\begin{equation}
\label{eq:1'}
F(\bx \vee \b1_{C'}) \geq t(1-t) f(C) + t f(C').
\end{equation}
Combining equations (\ref{eq:1}) and (\ref{eq:1'}), we get
\begin{eqnarray*}
& & F(\bx \vee \b1_{C'})  +  F(\bx \wedge \b1_C)
 \geq t (f(C) - f(C')) + t(1-t) f(C) + t \, f(C')
 = (2t - t^2) f(C).
\end{eqnarray*}
Therefore,
$$ F(\bx) \geq \frac12 (F(\bx \vee \b1_{C'}) + F(\bx \wedge \b1_C))
 \geq (t - \frac12 t^2) f(C).$$
Finally, we apply the pipage rounding technique which does not
lose anything in terms of objective value (see Lemma~\ref{lemma:pipage}).
\end{proof}

\paragraph{Technical remarks}
In each step of the algorithm, we need to estimate values of $F(\bx)$
for given $\bx \in P_t(\cM)$. We accomplish this by using the expression
$F(\bx) = \E[f(R(\bx))]$ where $R(\bx)$ is a random set associated with $\bx$.
By standard bounds, if the values of $f(S)$ are in a range $[0,M]$,
we can achieve accuracy $M / poly(n)$ using a polynomial number of samples.
We use the fact that $OPT \geq \frac{1}{n} M$ (see
Lemma~\ref{lemma:solution-value} in the Appendix) and therefore
we can achieve $OPT / poly(n)$ additive error in polynomial time.

We also relax the local step condition: we move to the next solution
only if $F(\bx + \delta \bv) > F(\bx) + \frac{\delta}{poly(n)} OPT$
for a suitable polynomial in $n$. This way, we can only make a polynomial number of steps.
When we terminate, the local optimality conditions (Lemma~\ref{lemma:local-opt})
are satisfied within an additive error of $OPT / poly(n)$,
which yields a polynomially small error in the approximation bound.

\subsection{Matroid base constraint}
\label{section:submod-bases}

Let us move on to the problem $\max \{f(S): S \in \cB\}$ where
$\cB$ are the bases of a matroid.
For a fixed $t \in [0,1]$, let us consider an algorithm which can be seen
as local search inside the base polytope $B(\cM)$, further constrained
by the box $[0,t]^X$. The matroid base polytope is defined as
$ B(\cM) = \mbox{conv} \{ \b1_B: B \in \cB \} $
or equivalently \cite{E70} as
$ B(\cM) = \{ \bx \geq 0: \forall S \subseteq X; \sum_{i \in S} x_i \leq r_\cM(S),
 \sum_{i \in X} x_i = r_\cM(X) \}, $
where $r_\cM$ is the matroid rank function of $\cM$. Finally, we define
$$ B_t(\cM) = B(\cM) \cap [0,t]^X = \{ \bx \in B(\cM): \forall i \in X; x_i \leq t \}.$$
Observe that $B_t(\cM)$ is nonempty if and only if there is a convex linear
combination $\bx = \sum_{B \in \cB} \xi_B \b1_B$ such that $x_i \in [0,t]$ for all $i$.
This is equivalent to saying that there is a linear combination (a fractional base packing)
$\bx' = \sum_{B \in \cB} \xi'_B \b1_B$ such that $x_i \in [0,1]$
and $\sum \xi'_B \geq \frac{1}{t}$, in other words the fractional base packing
number is $\nu \geq \frac{1}{t}$.  Since the optimal fractional packing of bases
in a matroid can be found efficiently (see Corollary 42.7a in \cite{Schrijver}, Volume B),
we can find efficiently the minimum $t \in [\frac12,1]$ such that $B_t(\cM) \neq \emptyset$.
Then, our algorithm is the following.

\

\noindent{\bf Fractional local search in $B_t(\cM)$} \\
 (given $t = \frac{r}{q}$, $r \leq q$ integer)
\begin{enumerate}
\item Let $\delta = \frac{1}{q}$.
Assume that $\bx \in B_t(\cM)$; adjust $\bx$ (using pipage rounding)
so that each $x_i$ is an integer multiple of $\delta$.
In the following, this property will be maintained.
\item If there is a direction $\bv = \be_j - \be_i$ such that
$\bx + \delta \bv \in B_t(\cM)$ and $F(\bx + \delta \bv) > F(\bx)$,
then set $\bx := \bx + \delta \bv$ and repeat.
\item If there is no such direction $\bv$, apply pipage rounding to $\bx$
and return the resulting solution.
\end{enumerate}

\noindent{\em Notes.}
We remark that the starting point can be found as a convex linear combination
of $q$ bases, $x = \frac{1}{q} \sum_{i=1}^{q} \b1_{B_i}$, such that
no element appears in more than $r$ of them, using matroid union techniques (see Theorem 42.9 in \cite{Schrijver}).
In the algorithm, we maintain this representation.
The local search step corresponds to switching a pair of elements in one base,
under the condition that no element is used in more than $r$ bases at the same time.
For now, we ignore the issues of estimating $F(\bx)$ and stopping the local search
within polynomial time. We discuss this at the end of this section.

Finally, we use pipage rounding to convert the fractional solution $\bx$
into an integral one of value at least $F(\bx)$ (Lemma~\ref{lemma:base-pipage} in the Appendix).
Note that it is not necessarily true that any of the bases in
a convex linear combination $\bx = \sum \xi_B \b1_{B}$ achieves the value $F(\bx)$.

\begin{theorem}
\label{thm:submod-bases-approx}
If there is a fractional packing of $\nu \in [1,2]$ bases in $\cM$,
then the fractional local search algorithm with $t = \frac{1}{\nu}$
returns a solution of value at least $\frac12 (1-t) \ OPT.$
\end{theorem}

For example, assume that $\cM$ contains two disjoint bases $B_1, B_2$
(which is the case considered in \cite{LMNS09}). Then,
the algorithm can be used with $t = \frac{1}{2}$ and
and we obtain a $(\frac{1}{4}-o(1))$-approximation, improving the
$(\frac{1}{6}-o(1))$-approximation from \cite{LMNS09}.
If there is a fractional packing of more than 2 bases,
our analysis still gives only a $(\frac{1}{4}-o(1))$-approximation.
If the dual matroid $\cM^*$ admits a better fractional packing of bases,
we can consider the problem $\max \{f(\bar{S}): S \in \cB^* \}$ which
is equivalent. For a uniform matroid, $\cB = \{B: |B|=k\}$,
the fractional base packing number is either at least $2$ or the same
holds for the dual matroid, $\cB^* = \{B: |B|=n-k\}$ (as noted in \cite{LMNS09}).
Therefore, we get a $(\frac{1}{4}-o(1))$-approximation for any uniform matroid.
The value $t=1$ can be used for any matroid,
but it does not yield any approximation guarantee.

\paragraph{Analysis of the algorithm}
We turn to the properties of fractional local optima.
We will prove that the point $\bx$ found by the fractional local search algorithm
satisfies the following conditions that allow us to compare $F(\bx)$ to the actual
optimum.

\begin{lemma}
\label{lemma:local-opt2}
The outcome of the fractional local search algorithm $\bx$
is a ``fractional local optimum'' in the following sense.
\begin{itemize}
\item For any $i,j$ such that $\bx + \delta (\be_j - \be_i) \in B_t(\cM)$,
$$\partdiff{F}{x_j} - \partdiff{F}{x_i} \leq 0.$$
(The partial derivatives are evaluated at $\bx$.)
\end{itemize}
\end{lemma}

\begin{proof}
Similarly to Lemma~\ref{lemma:local-opt},
observe that the coordinates of $\bx$ are always integer multiples of $\delta$,
therefore if it is possible to move from $\bx$ in the direction of
$\bv = \be_j - \be_i$ by any nonzero amount, then it is possible to
move by $\delta \bv$.
We use the property that for any direction $\bv = \be_j - \be_i$,
the function $F(\bx + \lambda \bv)$ is a convex function of $\lambda$ \cite{CCPV07}.
Therefore, if $\frac{dF}{d\lambda} > 0$ and it is possible
to move in the direction of $\bv$, we would get
$F(\bx + \delta \bv) > F(\bx)$ and the fractional local search would continue.
For $\bv = \be_j - \be_i$, we get
$$ \frac{dF}{d\lambda} = \partdiff{F}{x_j} - \partdiff{F}{x_i} \leq 0.$$
\end{proof}

We refer to the following exchange property for matroid bases
(see \cite{Schrijver}, Corollary 39.21a).

\begin{lemma}
\label{lemma:exchange2}
For any $B_1, B_2 \in \cB$, there is a bijection $\pi:B_1 \setminus B_2
 \rightarrow B_2 \setminus B_1$ such that
$\forall i \in B_1 \setminus B_2$; $B_1-i+\pi(i) \in \cM$.
\end{lemma}

Using this, we prove a lemma about fractional local optima
analogous to Lemma 2.2 in \cite{LMNS09}.

\begin{lemma}
\label{lemma:fractional-search2}
Let $\bx$ be the outcome of fractional local search over $B_t(\cM)$.
Let $C \in \cB$ be any base. Then there is $\bc \in [0,1]^X$ satisfying
\begin{itemize}
\item $c_i = t$, if $i \in C$ and $x_i = t$
\item $c_i = 1$, if $i \in C$ and $x_i < t$
\item $0 \leq c_i \leq x_i$, if $i \notin C$
\end{itemize}
such that
$$ 2 F(\bx) \geq F(\bx \vee \bc) + F(\bx \wedge \bc).$$
\end{lemma}

Note that for $t = 1$, we can set $\bc = \b1_C$.
However, in general we need this more complicated formulation.
Intuitively, $\bc$ is obtained from $\bx$ by raising the variables $x_i, i \in C$
and decreasing $x_i$ for $i \notin C$.
However, we can only raise the variables $x_i, i \in C$, where $x_i$ is below
the threshold $t$, otherwise we do not have any information about $\partdiff{F}{x_i}$.
Also, we do not necessarily decrease all the variables outside of $C$ to zero.

\begin{proof}
Let $C \in \cB$ and assume $\bx \in B_t(\cM)$ is a fractional local optimum.
We can decompose $\bx$ into a convex linear combination of vertices
of $B(\cM)$, $\bx = \sum \xi_B \b1_B$.
By Lemma~\ref{lemma:exchange2}, for each base $B$ there is a bijection
$\pi_B:B \setminus C \rightarrow C \setminus B$ such that
$\forall i \in B \setminus C$; $B - i + \pi_B(i) \in \cB$.

We define $C' = \{i \in C: x_i < t \}$.
The reason we consider $C'$ is that
if $x_i = t$, there is no room for an exchange step increasing $x_i$,
and therefore Lemma~\ref{lemma:local-opt2} does not give any information
about $\partdiff{F}{x_i}$.
We construct the vector $\bc$ by starting from $\bx$, and for each $B$
swapping the elements in $B \setminus C$ for their image under $\pi_B$,
provided it is in $C'$, until we raise the coordinates on $C'$ to $c_i=1$.
Formally, we set $c_i = 1$ for $i \in C'$, $c_i = t$ for $i \in C \setminus C'$,
and for each $i \notin C$, we define
$$ c_i = x_i - \sum_{B: i \in B, \pi_B(i) \in C'} \xi_B.$$

In the following, all partial derivatives are evaluated at $\bx$.
By the smooth submodularity of $F(\bx)$ (see \cite{CCPV09}),
\begin{eqnarray}
\label{eq:2}
F(\bx \vee \bc) - F(\bx) & \leq & \sum_{j: c_j > x_j} (c_j - x_j) \partdiff{F}{x_j}
 = \sum_{j \in C'} (1 - x_j) \partdiff{F}{x_j}
  = \sum_B \sum_{j \in C' \setminus B} \xi_B \partdiff{F}{x_j}
\end{eqnarray}
because $\sum_{B: j \notin B} \xi_B = 1 - x_j$ for any $j$.
On the other hand, also by smooth submodularity,
\begin{eqnarray*}
F(\bx) - F(\bx \wedge \bc) & \geq & \sum_{i: c_i < x_i} (x_i - c_i) \partdiff{F}{x_i}
 = \sum_{i \notin C} (x_i - c_i) \partdiff{F}{x_i}
 = \sum_{i \notin C} \sum_{B: i \in B, \pi_B(i) \in C'} \xi_B \partdiff{F}{x_i}
\end{eqnarray*}
using our definition of $c_i$. In the last sum,
for any nonzero contribution, we have $\xi_B > 0$, $i \in B$ and $j = \pi_B(i) \in C'$,
i.e. $x_j < t$. Therefore it is possible to move in the direction $\be_j - \be_i$
(we can switch from $B$ to $B-i+j$).
By Lemma~\ref{lemma:local-opt2},
$$ \partdiff{F}{x_j} - \partdiff{F}{x_i} \leq 0.$$
Therefore, we get
\begin{eqnarray}
F(\bx) - F(\bx \wedge \bc) & \geq &
\sum_{i \notin C} \sum_{B: i \in B, j=\pi_B(i) \in C'} \xi_B \partdiff{F}{x_j}
\label{eq:3}
 = \sum_B \sum_{i \in B \setminus C: j=\pi_B(i) \in C'} \xi_B \partdiff{F}{x_j}.
\end{eqnarray}
By the bijective property of $\pi_B$, this is equal to
$\sum_B \sum_{j \in C' \setminus B} \xi_B \partdiff{F}{x_j}$.
Putting (\ref{eq:2}) and (\ref{eq:3}) together, we get
$F(\bx \vee \bc) - F(\bx) \leq F(\bx) - F(\bx \wedge \bc)$.
\end{proof}

Now we are ready to prove Theorem~\ref{thm:submod-bases-approx}.

\begin{proof}
Assuming that $B_t(\cM) \neq \emptyset$, we can find a starting point $\bx_0 \in B_t(\cM)$.
From this point, we reach a fractional local optimum $\bx \in B_t(\cM)$
(see Lemma~\ref{lemma:local-opt2}).
We want to compare $F(\bx)$ to the actual optimum; assume that $OPT = f(C)$.

As before, we define $C' = \{i \in C: x_i < t\}$.
By Lemma~\ref{lemma:fractional-search2}, we know that
the fractional local optimum satisfies:
\begin{equation}
\label{eq:4}
2 F(\bx) \geq F(\bx \vee \bc) + F(\bx \wedge \bc)
\end{equation}
for some vector $\bc$ such that $c_i = t$ for all $i \in C \setminus C'$,
$c_i = 1$ for $i \in C'$ and $0 \leq c_i \leq x_i$ for $i \notin C$.

First, let's analyze $F(\bx \vee \bc)$. We have
\begin{itemize}
\item $(\bx \vee \bc)_i = 1$ for all $i \in C'$.
\item $(\bx \vee \bc)_i = t$ for all $i \in C \setminus C'$.
\item $(\bx \vee \bc)_i \leq t$ for all $i \notin C$.
\end{itemize}
We apply Lemma~\ref{lemma:submod-split} to the partition $X = C \cup \bar{C}$.
We get
$$ F(\bx \vee \bc) \geq \E[f((T_1(\bx \vee \bc) \cap C) \cup (T_2(\bx \vee \bc) \cap \bar{C}))] $$
where $T_1(\bx)$ and $T_2(\bx)$ are independent threshold sets.
Based on the information above, $T_1(\bx \vee \bc) \cap C = C$ with probability $t$
and $T_1(\bx \vee \bc) \cap C = C'$ otherwise. On the other hand,
$T_2(\bx \vee \bc) \cap \bar{C} = \emptyset$ with probability at least $1-t$.
These two events are independent. We conclude that on the right-hand side,
we get $f(C)$ with probability at least $t(1-t)$, or $f(C')$ with probability
at least $(1-t)^2$:
\begin{equation}
\label{eq:5}
F(\bx \vee \bc) \geq t(1-t) f(C) + (1-t)^2 f(C').
\end{equation}
Turning to $F(\bx \wedge \bc)$, we see that
\begin{itemize}
\item $(\bx \wedge \bc)_i = x_i$ for all $i \in C'$.
\item $(\bx \wedge \bc)_i = t$ for all $i \in C \setminus C'$.
\item $(\bx \wedge \bc)_i \leq t$ for all $i \notin C$.
\end{itemize}
We apply Lemma~\ref{lemma:submod-split} to $X = C \cup \bar{C}$.
$$ F(\bx \wedge \bc) \geq \E[f((T_1(\bx \wedge \bc) \cap C) \cup (T_2(\bx \wedge \bc) \cap \bar{C}))].$$
With probability $t$,  $T_1(\bx \wedge \bc) \cap C$ contains $C \setminus C'$
(and maybe some elements of $C'$). In this case, $f(T_1(\bx \wedge \bc) \cap C)
 \geq f(C) - f(C')$ by submodularity.
Also, $T_2(\bx \wedge \bc) \cap \bar{C}$ is empty with probability at least $1-t$.
Again, these two events are independent.
Therefore,
$ F(\bx \wedge \bc) \geq t(1-t) (f(C) - f(C')).$
If $f(C') > f(C)$, this bound is vacuous; otherwise, we can replace $t(1-t)$ by
$(1-t)^2$, because $t \geq 1/2$. In any case,
\begin{equation}
\label{eq:6}
F(\bx \wedge \bc) \geq (1-t)^2 (f(C) - f(C')).
\end{equation}
Combining (\ref{eq:4}), (\ref{eq:5}) and (\ref{eq:6}),
$$ F(\bx) \geq \frac12 (F(\bx \vee \bc) + F(\bx \wedge \bc))
 \geq \frac12 (t(1-t) f(C) + (1-t)^2 f(C)) = \frac12 (1-t) f(C).$$
\end{proof}

\paragraph{Technical remarks}
Again, we have to deal with the issues of estimating $F(\bx)$ and
stopping the local search in polynomial time. We do this exactly
as we did at the end of Section~\ref{section:submod-independent}.
One issue to be careful about here is that if $f:2^X \rightarrow [0,M]$,
our estimates of $F(\bx)$ are within an additive error of $M / poly(n)$.
If the optimum value $OPT = \max \{f(S): S \in \cB\}$ is very small
compared to $M$, the error might be large compared to $OPT$ which would be
a problem. The optimum could in fact be very small in general.
But it holds that if $\cM$ contains no loops and co-loops
(which can be eliminated easily), then $OPT \geq \frac{1}{n^2} M$
(see Appendix~\ref{app:base-value}). Then, our sampling errors are
on the order of $OPT / poly(n)$ which yields a $1/poly(n)$ error
in the approximation bound.

\section{Approximation for symmetric instances}
\label{section:symmetric}

We can achieve a better approximation assuming that the instance exhibits
a certain symmetry. This is the same kind of symmetry that we use in our hardness
construction (Section~\ref{section:hardness-proof}) and the hard instances exhibit
the same symmetry as well. It turns out that our approximation in this case
matches the hardness threshold up to lower order terms.

Similar to our hardness result, the symmetries that we consider
here are permutations of the ground set $X$, corresponding
to permutations of coordinates in $\RR^X$. We start with some basic
properties which are helpful in analyzing symmetric instances.

\begin{lemma}
\label{lemma:grad-sym}
Assume that $f:2^X \rightarrow \RR$ is invariant with respect to
a group of permutations $\cG$ and $F(\bx) = \E[f(\hat{\bx})]$.
Then for any symmetrized vector $\bar{\bc} = \E_{\sigma \in {\cG}}[\sigma(\bc)]$,
$\nabla F |_{\bar{\bc}}$ is also symmetric w.r.t. $\cG$. I.e., for any $\tau \in \cG$,
$$ \tau(\nabla F |_{\bx=\bar{\bc}}) = \nabla F |_{\bx=\bar{\bc}}. $$
\end{lemma}

\begin{proof}
Since $f(S)$ is invariant under $\cG$, so is $F(\bx)$,
i.e. $F(\bx) = F(\tau(\bx))$ for any $\tau \in {\cal G}$.
Differentiating both sides at $\bx=\bc$, we get by the chain rule:
$$ \partdiff{F}{x_i} \Big|_{\bx=\bc} =
 \sum_j \partdiff{F}{x_j} \Big|_{\bx=\tau(\bc)} \partdiff{}{x_i} (\tau(x))_j
 = \sum_j \partdiff{F}{x_j} \Big|_{\bx=\tau(\bc)} \partdiff{x_{\tau(j)}}{x_i}. $$
Here, $\partdiff{x_{\tau(j)}}{x_i} = 1$ if $\tau(j) = i$, and $0$ otherwise.
Therefore,
$$ \partdiff{F}{x_i} \Big|_{\bx=\bc} =
 \partdiff{F}{x_{\tau^{-1}(i)}} \Big|_{\bx=\tau(\bc)}. $$
Note that $\tau(\bar{\bc}) = \E_{\sigma \in \cG}[\tau(\sigma(\bc))]
 =  \E_{\sigma \in \cG}[\sigma(\bc)] = \bar{\bc}$
since the distribution of $\tau \circ \sigma$ is equal to the distribution
of $\sigma$. Therefore,
$$ \partdiff{F}{x_i} \Big|_{\bx=\bar{\bc}} = \partdiff{F}{x_{\tau^{-1}(i)}}
  \Big|_{\bx=\bar{\bc}} $$
for any $\tau \in \cG$.
\end{proof}

Next, we prove that the ``symmetric optimum''
$\max \{F(\bar{\bx}): \bx \in P(\cF)\}$ gives a solution which is
a local optimum for the original instance $\max \{F(\bx): \bx \in P(\cF)\}$.
(As we proved in Section~\ref{section:hardness-proof}, in general
we cannot hope to find a better solution than the symmetric optimum.)

\begin{lemma}
\label{lemma:sym-local-opt}
Let $f:2^X \rightarrow \RR$ and ${\cal F} \subset 2^X$
be invariant with respect to a group of permutations $\cal G$.
Let $\overline{OPT} = \max \{ F(\bar{\bx}): \bx \in P(\cF) \}$
where $\bar{\bx} = \E_{\sigma \in {\cal G}}[\sigma(\bx)]$,
and let $\bx_0$ be the symmetric point
where $\overline{OPT}$ is attained  ($\bar{\bx}_0 = \bx_0$).
Then $\bx_0$ is a local optimum for the problem
$\max \{F(\bx): \bx \in P(\cF) \}$, in the sense that
$(\bx-\bx_0) \cdot \nabla F |_{\bx_0} \leq 0$ for any $\bx \in P(\cF)$.
\end{lemma}

\begin{proof}
Assume for the sake of contradiction that $(\bx-\bx_0) \cdot \nabla F|_{\bx_0} > 0$
for some $\bx \in P(\cF)$. We use the symmetric properties of $f$ and $\cF$
to show that $(\bar{\bx} - \bx_0) \cdot \nabla F|_{\bx_0} > 0$ as well.
Recall that $\bx_0 = \bar{\bx}_0$. We have
$$ (\bar{\bx} - \bx_0) \cdot \nabla F|_{\bx_0}
 = \E_{\sigma \in \cG}[\sigma(\bx-\bx_0) \cdot \nabla F|_{\bx_0}] 
 = \E_{\sigma \in \cG}[(\bx-\bx_0) \cdot \sigma^{-1}(\nabla F|_{\bx_0})]
 = (\bx-\bx_0) \cdot \nabla F|_{\bx_0} > 0 $$
using Lemma~\ref{lemma:grad-sym}.  Hence, there would be a direction
$\bar{\bx} - \bx_0$ along which an improvement can be obtained.
But then, consider a small $\delta > 0$ such that
$\bx_1 = \bx_0 + \delta (\bar{\bx} - \bx_0) \in P(\cF)$ and also
$F(\bx_1) > F(\bx_0)$. The point $\bx_1$ is symmetric ($\bar{\bx}_1 = \bx_1$)
and hence it would contradict the assumption that $F(\bx_0) = \overline{OPT}$.
\end{proof}

\subsection{Submodular maximization over independent sets in a matroid}

Let us derive an optimal approximation result for the problem
$\max \{f(S): S \in \cI\}$ under the assumption that the instance is "element-transitive". 

\begin{definition}
For a group $\cG$ of permutations on $X$, the orbit of an element $i \in X$
is the set $\{ \sigma(i): \sigma \in \cG \}$.
$\cG$ is called element-transitive, if the orbit of any element is
the entire ground set $X$.
\end{definition}

In this case, we show that it is easy to achieve an optimal $(\frac12-o(1))$-approximation
for a matroid independence constraint.

\begin{theorem}
\label{thm:sym-submod-independent}
Let $\max \{f(S): S \in \cI\}$ be an instance symmetric with respect
to an element-transitive group of permutations $\cG$. Let
$\overline{OPT} = \max \{F(\bar{\bx}): \bx \in P(\cM)\}$
 where $\bar{\bx} = \E_{\sigma \in \cG}[\sigma(\bx)]$.
Then $\overline{OPT} \geq \frac12 OPT$.
\end{theorem}

\begin{proof}
Let $OPT = f(C)$.
By Lemma~\ref{lemma:sym-local-opt}, $\overline{OPT} = F(\bx_0)$
where $\bx_0$ is a local optimum for the problem $\max \{F(\bx): \bx \in P(\cM)$.
This means it is also a local optimum in the sense of Lemma~\ref{lemma:local-opt},
with $t=1$.
By Lemma~\ref{lemma:fractional-search},
$$ 2 F(\bx_0) \geq F(\bx_0 \vee \b1_C) + F(\bx_0 \wedge \b1_C).$$
Also, $\bx_0 = \bar{\bx}_0$. As we are dealing with an element-transitive
group of symmetries, this means all the coordinates of $\bx_0$ are equal,
$\bx_0 = (\xi,\xi,\ldots,\xi)$. Therefore, $\bx_0 \vee \b1_C$ is equal
to $1$ on $C$ and $\xi$ outside of $C$. By Lemma~\ref{lemma:rnd-threshold} (in the Appendix),
$$ F(\bx_0 \vee \b1_C) \geq (1-\xi) f(C).$$
Similarly, $\bx_0 \wedge \b1_C$ is equal to $\xi$ on $C$ and $0$ outside of $C$.
By Lemma~\ref{lemma:rnd-threshold},
$$ F(\bx_0 \wedge \b1_C) \geq \xi f(C).$$
Combining the two bounds,
$$ 2 F(\bx_0) \geq F(\bx_0 \vee \b1_C) + F(\bx_0 \wedge \b1_C) 
 \geq (1-\xi) f(C) + \xi f(C) = f(C) = OPT.$$
\end{proof}

Since all symmetric solutions $\bx = (\xi,\xi,\ldots,\xi)$ form
a 1-parameter family, and $F(\xi,\xi,\ldots,\xi)$ is a concave function,
we can search for the best symmetric solution (within any desired accuracy)
by binary search. By standard techniques, we get the following.

\begin{corollary}
There is a $(\frac12 - o(1))$-approximation ("brute force" search
over symmetric solutions) for the problem $\max \{f(S): S \in \cI\}$
for instances symmetric under an element-transitive group of permutations.
\end{corollary}

The hard instances for submodular maximization subject to a matroid
independence constraint correspond to refinements of the
Max Cut instance for the graph $K_2$ (Section~\ref{section:hardness-applications}).
It is easy to see that such instances are element-transitive,
and it follows from Section~\ref{section:hardness-proof} that a $(\frac12+\epsilon)$-approximation
for such instances would require exponentially many value queries.
Therefore, our approximation for element-transitive instances is optimal.

\subsection{Submodular maximization over bases}

Let us come back to the problem of submodular maximization over the bases
of matroid. The property that $\overline{OPT}$ is a local optimum
with respect to the original problem $\max \{F(\bx): \bx \in P(\cF)\}$
is very useful in arguing about the value of $\overline{OPT}$.
We already have tools to deal with local optima from
Section~\ref{section:submod-bases}. Here we prove the following.

\begin{lemma}
\label{lemma:base-local-opt}
Let $B(\cM)$ be the matroid base polytope of $\cM$ and
$\bx_0 \in B(\cM)$ a local maximum for the submodular maximization problem
$\max \{F(\bx): \bx \in B(\cM)\}$, in the sense that $(\bx-\bx_0) \cdot \nabla F|_{\bx_0} \leq 0$
for any $\bx \in B(\cM)$. Assume in addition that $\bx_0 \in [s,t]^X$. Then
$$ F(\bx_0) \geq \frac12 (1-t+s) \cdot OPT.$$
\end{lemma}

\begin{proof}
Let $OPT = \max \{ f(B): B \in \cB \} = f(C)$.
We assume that $\bx_0 \in B(\cM)$ is a local optimum with respect
to any direction $\bx-\bx_0$, $\bx \in B(\cM)$, so it is also a local optimum with respect
to the fractional local search in the sense of Lemma~\ref{lemma:fractional-search2},
with $t=1$. The lemma implies that
$$ 2 F(\bx_0) \geq F(\bx \vee \b1_C) + F(\bx \wedge \b1_C).$$
By assumption, the coordinates of $\bx \vee \b1_C$ are equal to $1$ on $C$
and at most $t$ outside of $C$. With probability $1-t$, a random threshold in $[0,1]$
falls between $t$ and $1$, and Lemma~\ref{lemma:rnd-threshold} (in the Appendix) implies that
$$ F(\bx \vee \b1_C) \geq (1-t) \cdot f(C).$$
Similarly, the coordinates of $\bx \wedge \b1_C$ are $0$ outside of $C$,
and at least $s$ on $C$. A random threshold falls between $0$ and $s$
with probability $s$, and Lemma~\ref{lemma:rnd-threshold} implies that
$$ F(\bx \wedge \b1_C) \geq s \cdot f(C).$$
Putting these inequalities together, we get
$$2 F(\bx_0) \geq F(\bx \vee \b1_C) + F(\bx \wedge \b1_C) \geq (1-t+s) \cdot f(C).$$
\end{proof}

\

\noindent
{\bf Totally symmetric instances.}
The application we have in mind here is a special case of submodular
maximization over the bases of a matroid, which we call {\em totally symmetric}.

\begin{definition}
\label{def:totally-symmetric}
We call an instance $\max \{f(S): S \in \cF \}$ totally symmetric with respect
to a group of permutations $\cG$, if both $f(S)$ and $\cF$ are invariant
under $\cG$ and moreover, there is a point $\bc \in P(\cF)$ such that
$\bc = \bar{\bx} = \E_{\sigma \in \cG}[\sigma(\bx)]$ for every $\bx \in P(\cF)$.
We call $\bc$ the center of the instance.
\end{definition}

Note that this is indeed stronger than just being invariant under $\cG$.
For example, an instance on a ground set $X = X_1 \cup X_2$ could be symmetric
with respect to any permutation of $X_1$ and any permutation of $X_2$.
For any $\bx \in P(\cF)$, the symmetric vector $\bar{\bx}$ is constant
on $X_1$ and constant on $X_2$. However, in a totally symmetric instance,
there should be a unique symmetric point.

\paragraph{Bases of partition matroids}
A canonical example of a totally symmetric instance is as follows.
Let $X = X_1 \cup X_2 \cup \ldots \cup X_m$ and let integers $k_1,\ldots,k_m$
be given. This defines a partition matroid $\cM = (X,\cB)$, whose bases are
$$ \cB = \{B: \forall j; |B \cap X_j| = k_j \}.$$
The associated matroid base polytope is
$$ B(\cM) = \{\bx \geq 0: \forall j; \sum_{i \in X_j} x_i = k_j \}.$$
Let $\cG$ be a group of permutations such that the orbit of each element
$i \in X_j$ is the entire set $X_j$.
This implies that for any $\bx \in B(\cM)$,
$\bar{\bx}$ is the same vector $\bc$, with coordinates $k_j / |X_j|$ on $X_j$.
If $f(S)$ is also invariant under $\cG$, we have a totally symmetric
instance $\max \{f(S): S \in \cB \}$. 

\paragraph{Example: welfare maximization}
To present a more concrete example, consider $X_j = \{ a_{j1}, \ldots, a_{jk} \}$
for each $j \in [m]$, a set of bases $\cB = \{B: \forall j; |B \cap X_j| = 1\}$,
and an objective function in the form $f(S) = \sum_{i=1}^{k} v(\{j: a_{ji} \in S\})$,
where $v:2^{[m]} \rightarrow \RR_+$ is a submodular function.
This is a totally symmetric instance, which captures the welfare maximization
problem for combinatorial auctions where each player has the same valuation function.
(Including element $a_{ji}$ in the solution corresponds to allocating item $j$ to player $i$;
see \cite{CCPV09} for more details.) We remark that here we consider possibly
nonmonotone submodular functions, which is not common for combinatorial auctions;
nevertheless the problem still makes sense.

\

We show that for such instances, the center point achieves an improved approximation.

\begin{theorem}
\label{thm:sym-solution}
Let $\max \{f(S): S \in \cB\}$ be a totally symmetric instance.
Let the fractional packing number of bases be $\nu$ and the fractional
packing number of dual bases $\nu^*$. Then the center point $\bc$
satisfies
$$ F(\bc) \geq \left(1 - \frac{1}{2\nu} - \frac{1}{2\nu^*} \right) OPT.$$
\end{theorem}

Recall that in the general case, we get a $\frac12 (1-1/\nu-o(1))$-approximation
(Theorem~\ref{thm:submod-bases-approx}). By passing to the dual matroid,
we can also obtain a $\frac12 (1-1/\nu^*-o(1))$-approximation,
so in general, we know how to achieve a $\frac12 (1-1/\max\{\nu,\nu^*\}-o(1))$-approximation.
For totally symmetric instances where $\nu=\nu^*$, we improve this
to the optimal factor of $1-1/\nu$.

\begin{proof}
Since there is a unique center $\bc = \bar{\bx}$ for any $\bx \in B(\cM)$,
this means this is also the symmetric optimum $F(\bc) = \max \{ F(\bar{\bx}): \bx \in B(\cM)\}$.
Due to Lemma~\ref{lemma:sym-local-opt}, $\bc$ is a local optimum
for the problem $\max \{F(\bx): \bx \in B(\cM)\}$.

Because the fractional packing number of bases is $\nu$, we have $c_i \leq 1/\nu$
for all $i$. Similarly, because the fractional packing number of dual bases
(complements of bases) is $\nu^*$,
we have $1-c_i \leq 1/\nu^*$. This means that $\bc \in [1-1/\nu^*,1/\nu]$.
Lemma~\ref{lemma:base-local-opt} implies that
$$ 2 F(\bc) \geq \left( 1-\frac{1}{\nu} + 1-\frac{1}{\nu^*} \right) OPT.$$
\end{proof}

\begin{corollary}
Let $\max \{f(S): S \in \cB\}$ be an instance
on a partition matroid where every base takes at least
an $\alpha$-fraction of each part, at most a $(1-\alpha)$-fraction of each part,
and the submodular function $f(S)$ is invariant under a group $\cG$
where the orbit of each $i \in X_j$ is $X_j$.
Then, the center point $\bc = \E_{\sigma \in \cG}[\sigma(\b1_B)]$
(equal for any $B \in \cB$) satisfies $F(\bc) \geq \alpha \cdot OPT$.
\end{corollary}

\begin{proof}
If the orbit of any element $i \in X_j$ is the entire set $X_j$,
it also means that $\sigma(i)$ for a random $\sigma \in \cG$
is uniformly distributed over $X_j$ (by the transitive property of $\cG$).
Therefore, symmetrizing any fractional vector $\bx \in B(\cM)$ gives
the same vector $\bar{\bx} = \bc$, where $c_i = k_j / |X_j|$ for $i \in X_j$.
Also, our assumptions mean that the fractional packing number of bases
is $1/(1-\alpha)$, and the fractional packing number of dual bases is
also $1/(1-\alpha)$.
Due to Lemma~\ref{thm:sym-solution}, the center $\bc$ satisfies
$F(\bc) \geq \alpha \cdot OPT$.
\end{proof}

The hard instances for submodular maximization over matroid bases
that we describe in Section~\ref{section:hardness-applications} are exactly of this form
(see the last paragraph of Section~\ref{section:hardness-applications}, with $\alpha=1/k$).
There is a unique symmetric solution,
$x = (\alpha,\alpha,\ldots,\alpha, 1-\alpha,1-\alpha,\ldots,1-\alpha)$.
The fractional base packing number for these matroids is
$\nu = 1/(1-\alpha)$ and Theorem~\ref{thm:general-hardness} implies that
any $(\alpha+\epsilon) = (1-1/\nu+\epsilon)$-approximation for such matroids would
require exponentially many value queries. Therefore, our approximation
in this special case is optimal.

\section*{Acknowledgment}
The author would like to thank Jon Lee and Maxim Sviridenko for helpful discussions.

\appendix

\section{Submodular functions and their extensions}

In this appendix, we present a few basic facts concerning
submodular functions $f(S)$ and their continuous extensions.
By $f_A(S)$, we denote the marginal value of $S$ w.r.t. $A$,
$f_A(S) = f(A \cup S) - f(A)$. We also use $f_A(i) = f(A+i) - f(A)$.
The notation $A+i$ is shorthand for $A \cup \{i\}$.
Similarly, we write $A-i$ to denote $A \setminus \{i\}$.

\begin{definition}
\label{def:multilinear}
The multilinear extension of a function $f:2^X \rightarrow \RR$
is a function $F:[0,1]^X \rightarrow \RR$ where
$$ F(\bx) = \sum_{S \subseteq X} f(S) \prod_{i \in S} x_i \prod_{j \notin S} (1-x_j).$$
\end{definition}

\begin{definition}
\label{def:Lovasz}
The Lov\'asz extension of a function $f:2^X \rightarrow \RR$ is
a function $\tilde{F}:[0,1]^X \rightarrow \RR$ such that
$$ \tilde{F}(\bx) = \sum_{i=0}^{n} (x_{\pi(i)} - x_{\pi(i+1)})
 f(\{\pi(j): 1 \leq j \leq i\}) $$
where $\pi:[n] \rightarrow X$ is a bijection such that
$x_{\pi(1)} \geq x_{\pi(2)} \geq \ldots \geq x_{\pi(n)}$;
we interpret $x_{\pi(0)}$, $x_{\pi(n+1)}$ as $x_{\pi(0)} = 1, x_{\pi(n+1)} = 0$.
\end{definition}

A useful way to view the multilinear extension is that we sample a random
set $R(\bx)$, where each element $i$ appears independently with probability $x_i$,
and we take $F(\bx) = \E[f(R(\bx))]$. The Lov\'asz extension can be viewed
similarly, by sampling a random set in a correlated fashion.

\begin{definition}
\label{def:threshold}
For a vector $\bx \in [0,1]^X$, we define the ``random threshold set''
$T(\bx)$ by taking a uniformly random $\lambda \in [0,1]$, and setting
$$ T(\bx) = \{i \in X: x_i > \lambda \}.$$
\end{definition}

Assuming that $x_1 \geq x_2 \geq \ldots \geq x_n$, $x_0=1$, $x_n=0$,
it is easy to see that the Lov\'asz extension of $f(S)$ is equal to
$$ \tilde{F}(\bx) = \sum_{i=0}^{n} (x_i-x_{i+1}) f([i]) = \E[f(T(\bx))].$$
It is known that the Lov\'asz extension of a submodular function
is always convex \cite{L83}, which is not true for the multilinear extension.
We prove that the Lov\'asz extension is always upper-bounded by the
multilinear extension; this lemma appears quite basic but it has not
been published to our knowledge.

\begin{lemma}
\label{lemma:rnd-threshold}
Let $F(\bx)$ denote the multilinear extension and let
$\tilde{F}(\bx)$ denote the Lov\'asz extension of a submodular function
$f:2^X \rightarrow \RR$.
Then
$$ F(\bx) \geq \tilde{F}(\bx).$$
\end{lemma}

\begin{proof}
Let $R(\bx)$ be a random set where elements are sampled independently with probabilities
$x_i$, and let $T(\bx)$ be the random threshold set.
I.e., we want to prove $\E[f(R(\bx))] \geq \E[f(T(\bx))]$.
We can assume WLOG that the elements are ordered so that
 $x_1 \geq x_2 \geq \ldots \geq x_n$.
We also let $x_0 = 1$ and $x_{n+1} = 0$. Then,
\begin{eqnarray*}
\E[f(R(\bx))] & = & f(\emptyset) + \sum_{k=1}^{n} \E[f(R(\bx) \cap [k]) - f(R(\bx) \cap [k-1])] \\
  & = & f(\emptyset) + \sum_{k=1}^{n} \E[f_{R(\bx) \cap [k-1]}(R(\bx) \cap \{k\})] 
\end{eqnarray*}
and by submodularity,
\begin{eqnarray*}
\E[f(R(\bx))] & \geq & f(\emptyset) + \sum_{k=1}^{n} \E[f_{[k-1]}(R(\bx) \cap \{k\})] \\
 & = & f(\emptyset) + \sum_{k=1}^{n} x_k f_{[k-1]}(k) \\
 & = & f(\emptyset) + \sum_{k=1}^{n} x_k (f([k]) - f([k-1])) \\
 & = & \sum_{k=0}^{n} (x_k - x_{k+1}) f([k])  \\
 & = & \E[f(T(\bx))].
\end{eqnarray*}
\end{proof}

A refinement of this lemma says that we can also consider a partition
of the ground set and apply an independent threshold set on each part.
This gives a certain hybrid between the multilinear and Lov\'asz extensions.

\begin{lemma}
\label{lemma:submod-split}
For any partition $X = X_1 \cup X_2$,
$$ F(\bx) \geq \E[f((T_1(\bx) \cap X_1) \cup (T_2(\bx) \cap X_2))] $$
where $T_1(\bx)$ and $T_2(\bx)$ are two independently random threshold sets for $\bx$.
\end{lemma}

Intuitively, sampling from $Y$ independently with probability $p$ is not worse
than taking all of $Y$ with probability $p$ or none of it with probability $1-p$.
This follows from statements proved in \cite{FMV07}, but we give a proof here
for completeness.

\begin{proof}
$F(\bx) = \E[f(R(\bx))]$ where $R(\bx)$ is sampled independently with probabilities $x_i$.
Let's condition on $R(\bx) \cap X_2 = R_2$. This restricts the remaining randomness
to $X_1$, and we get $f(R(\bx)) = f(R_2) + g(R(\bx))$, where $g(S) = f_{R_2}(S \cap X_1)$
is again submodular. By Lemma~\ref{lemma:rnd-threshold},
$\E[g(R(\bx))] \geq \E[g(T_1(\bx))]$ where $T_1(\bx)$ is a random threshold set.
Hence,
$$ \E[f(R(\bx)) \mid R(\bx) \cap X_2 = R_2] = f(R_2) + \E[g(R(\bx))] $$
$$ \geq f(R_2) + \E[g(T_1(\bx))] = \E[f(R_2 \cup (T_1(\bx) \cap X_1))].$$
By randomizing $R_2 = R(\bx) \cap X_2$, we get
$$ \E[f(R(\bx))] \geq \E[f((R(\bx) \cap X_2) \cup (T_1(\bx) \cap X_1))].$$
Repeating the same process, conditioning on $T_1(\bx)$ and applying
Lemma~\ref{lemma:rnd-threshold} to $R(\bx) \cap X_2$, we get
$$ \E[f(R(\bx))] \geq \E[f((T_2(\bx) \cap X_2) \cup (T_1(\bx) \cap X_1))].$$
\end{proof}

\section{Pipage rounding for nonmonotone submodular functions}
\label{app:pipage}

The pipage rounding technique \cite{AS04,CCPV07,CCPV09}
starts with a point in the base
polytope $\by \in B(\cM)$ and produces an integral solution $S \in \cI$
(in fact, a base) of expected value $\E[f(S)] \geq F(\by)$.
We recall the procedure here, in its randomized form \cite{CCPV09}.

\vspace{-10pt}
\begin{tabbing}
\ \ \ \= \ \ \= \ \ \= \ \ \= \ \ \= \ \ \  \\
{\em Subroutine} {\bf HitConstraint}($\by$, $i$, $j$): \\
\> Denote ${\cal A} = \{A \subseteq X: i \in A, j \notin A \}$; \\
\> Find $\delta = \min_{A \in {\cal A}} (r_\cM(A)-y(A))$ \\
\> \ \ \ \ \ and an optimal $A \in {\cal A}$; \\
\> If $y_j < \delta$ then ~$\{ \delta \leftarrow y_j,
  \ A \leftarrow \{j\} \}$;\\
\> $y_i \leftarrow y_i + \delta, \ y_j \leftarrow y_j - \delta$; \\
\> Return $(\by,A)$. \\
\> \\
{\em Algorithm} {\bf PipageRound}($\cM,\by$): \\
\> While ($\by$ is not integral) do \\
\> \> $T \leftarrow X$;  \\
\> \> While ($T$ contains fractional variables) do \\
\> \> \> Pick $i,j \in T$ fractional; \\
\> \> \> $(\by^+, A^+) \leftarrow {\bf HitConstraint}(\by,i,j)$; \\
\> \> \> $(\by^-, A^-) \leftarrow {\bf HitConstraint}(\by,j,i)$; \\
\> \> \> $p \leftarrow ||\by^+ - \by|| / ||\by^+ - \by^-||$; \\
\> \> \> With probability $p$, ~$\{ \by \leftarrow \by^-$,  $T \leftarrow T \cap A^- \}$; \\
\> \> \> ~~~~~~~~~~~~~~~~~~~Else ~~$\{ \by \leftarrow \by^+$, $T \leftarrow T \cap A^+ \}$;\\
\> \> EndWhile \\
\> EndWhile \\
\> Output $\by$.
\end{tabbing}

The application
in \cite{CCPV07} was to monotone submodular functions, but as the authors
mention, monotonicity is not used anywhere in the analysis. The technique
as described in \cite{CCPV09} yields the following.

\begin{lemma}
\label{lemma:base-pipage}
The pipage rounding technique,
given a membership oracle for a matroid $\cM = (X,\cI)$, a value oracle
for a submodular function $f:2^X \rightarrow \RR_+$, and $\by$ in the base
polytope $B(\cM)$, returns a random base $B$ of value $\E[f(B)] \geq F(\by)$.
\end{lemma}

Monotonicity is used in \cite{CCPV07} only to argue that a fractional solution
can be assumed to lie in the base polytope without loss of generality.
Therefore, if we are working with the base polytope (as in
 Section~\ref{section:submod-bases}), we can use the pipage rounding technique
without any modification.

If we are working with nonmonotone submodular functions and the matroid polytope,
we have to do some additional adjustments to make sure that we do not lose
anything when rounding a fractional solution. We proceed as follows.
The following procedure takes a point $\bx \in P(\cM)$ and while
there is a fractional coordinate, it either pushes it to its maximum possible value,
or makes it zero and removes it from the problem.

\vspace{-10pt}
\begin{tabbing}
\ \ \ \= \ \ \= \ \ \= \ \ \= \ \ \= \ \ \  \\
{\em Algorithm} {\bf Adjust}($\cM,\bx$): \\
\> While ($\bx \notin B(\cM)$) do \\
\> \> If (there is $i$ and $\delta>0$ such that $\bx + \delta \be_i \in P(\cM)$) do \\
\> \> \> Let $x_{max} = x_i + \max \{ \delta: \bx+\delta \be_i \in P(\cM) \}$; \\
\> \> \> Let $p = x_i / x_{max}$; \\
\> \> \> With probability $p$, ~$\{ x_i \leftarrow x_{max} \}$;\\
\> \> \> ~~~~~~~~~~~~~~~~~~~Else ~~$\{ x_i \leftarrow 0 \}$;\\
\> \> EndIf \\
\> \> If (there is $i$ such that $x_i = 0$) do \\
\> \> \> Delete $i$ from $\cM$ and remove the $i$-coordinate from $x$. \\
\> EndWhile \\
\> Output $(\cM,\bx)$.
\end{tabbing}

\begin{lemma}
\label{lemma:adjust}
Given $\bx \in P(\cM)$ and a submodular function $f(S)$ with its extension $F(\bx)$,
the procedure {\bf Adjust}($\cM,\bx$) yields a restricted matroid $\cM'$ and a point
$\by \in B(\cM')$ such that $\E[F(\by)] \geq F(\bx)$.
\end{lemma}

\begin{proof}
If $\bx \in P(\cM)$, there is always a point $\bz \in B(\cM)$ dominating $\bx$
in every coordinate (see Corollary 40.2h in \cite{Schrijver}).
Therefore, if $\bx \notin B(\cM)$, there is a coordinate which can be increased
while staying in $P(\cM)$. In each step, when we choose randomly between increasing
and decreasing $x_i$, the objective function is linear and the expectation
of $F(\bx)$ is preserved. Hence, the process forms a martingale. At the end,
$\E[F(\by)]$ is equal to the initial value $F(\bx)$.

As long as we increase variables, each variable is increased to its
maximum value and cannot be increased twice. Therefore, after at most $n$ steps
we either reach a point in $B(\cM)$ or make some variable $x_i$ equal to $0$.
A variable equal to $0$ is removed, which can be repeated at most $n$ times.
Hence, the process terminates in $O(n^2)$ time.
\end{proof}

For a given $\bx \in P(\cM)$, we run $(\cM',\by):=\mbox{\bf Adjust}(\cM,\bx)$,
followed by {\bf PipageRound} ($\cM',\by$). The outcome is a base in the restricted
matroid where some elements have been deleted, i.e. an independent set
in the original matroid. We call this the {\em extended pipage rounding} procedure.
Lemma~\ref{lemma:adjust} together with Lemma~\ref{lemma:base-pipage} gives
the following.

\begin{lemma}
\label{lemma:pipage}
The extended pipage rounding procedure,
given a membership oracle for a matroid $\cM = (X,\cI)$, a value oracle
for a submodular function $f:2^X \rightarrow \RR_+$, and $\bx$ in the matroid
polytope $P(\cM)$, yields a random independent set $S \in \cI$ of value
$\E[f(S)] \geq F(\bx)$.
\end{lemma}

\section{Solutions of nontrivial value}
\label{app:base-value}

In this section, we prove that solutions of nontrivial value always
exist for the problems of maximizing a nonmonotone submodular function
subject to a matroid independence or matroid base constraint.
We can assume that our matroid does not contain any loops
(which can be removed beforehand), and in the case of a matroid base
constraint no co-loops either (since co-loops participate in every solution
and can be contracted). The bounds that we prove here are needed to ensure
that our sampling errors can be made negligible compared to the value
of the optimum.

\begin{lemma}
\label{lemma:solution-value}
Let $\cM = (X, \cI)$ be a matroid without loops
(elements which are never in an independent set), $|X|=n$,
and let $f:2^X \rightarrow \RR_+$ be a (nonmonotone) submodular function
such that $\max_{S \subseteq X} f(S) = M$. Let $OPT = \max \{f(I): I \in \cI\}$.
Then
$$ OPT \geq \frac{1}{n} M.$$
\end{lemma}

\begin{proof}
We consider only solutions of size $|I| \leq 1$, which are independent
because our matroid does not contain any loops.
Let $f(I^*) = \max \{f(I): |I| \leq 1 \}$. By submodularity and nonnegativity
of $f$, the value of any nonempty set can be estimated as
$$ f(S) \leq \sum_{j \in S} f(\{j\}) \leq n \, f(I^*).$$
By definition, we also have $f(\emptyset) \leq f(I^*)$.
Therefore, $M = \max f(S) \leq n \, f(I^*)$.
\end{proof}

\begin{lemma}
\label{lemma:base-value}
Let $\cM = (X, \cI)$ be a matroid without loops (elements which are never in a base)
and coloops (elements which are in every base), $|X|=n$ and let $\cB$
denote the bases of $\cM$. Let $f:2^X \rightarrow \RR_+$ be a (nonmonotone)
submodular function such that $\max_{S \subseteq X} f(S) = M$. Let $OPT = \max \{f(B): B \in \cB\}$.
Then
$$ OPT \geq \frac{1}{n^2} M.$$
\end{lemma}

\begin{proof}
Let us find $B = \{b_1,\ldots,b_k \}$ greedily, by including in each step
an element which has maximum possible marginal value
 $f_{\{b_1,\ldots,b_\ell\}}(b_{\ell+1})$
among all elements that preserve independence.
We claim that $f(B) \geq \frac{1}{n^2} M$.

The first element $b_1$ is the element of maximum value $\max_{i \in X} f(\{i\})$
(because there are no loops, all elements are eligible).
Let $M = f(S^*)$. By submodularity,
$$ f(S^*) \leq \sum_{i \in S^*} f(\{i\}) \leq n f(\{b_1\}).$$
Therefore, $f(\{b_1\}) \geq \frac{1}{n} M$. However, additional elements
can decrease this value.
Consider the last element $b_k$ that we added to $B$ and assume that
$f(B) < \frac{1}{n^2} M$. Due to our greedy procedure and submodularity,
the marginal values of successive elements keep decreasing,
and therefore
$$f_{B\setminus \{b_k\}}(b_k) \leq \frac{1}{n-1} (f(B)-f(\{b_1\}))
 < -\frac{1}{n^2} M.$$
Since there are no coloops in $\cM$, adding $b_k$ was not our only choice
 - otherwise, we would have $\mbox{rank}(X \setminus \{b_k\}) < \mbox{rank}(X)$
which means exactly that $b_k$ is a coloop. So there is another element
$b'_k$ that would form a base $\{b_1,\ldots,b_{k-1},b'_k\}$.
The reason we did not select $b'_k$
was that $f_{B \setminus \{b_k\}}(b'_k) \leq f_{B \setminus \{b_k\}}(b_k)
 < -\frac{1}{n^2} M$.
By submodularity, if we add both elements, we obtain a set
$\tilde{B} = \{b_1,\ldots,b_k,b'_k\}$ such that
$$ f(\tilde{B}) = f(B) + f_B(b'_k)  \leq f(B) + f_{B \setminus \{b_k\}}(b'_k)
 < \frac{1}{n^2} M - \frac{1}{n^2} M = 0 $$
which is a contradiction with the nonnegativity of $f$.
\end{proof}

We remark that this bound is tight up to a constant factor.
Consider a complete bipartite directed graph
on $X = X_1 \cup X_2$, $|X_1| = |X_2| = n/2$. The submodular function $f(S)$ is
the directed cut function, expressing the number of arcs from $S$ to $\bar{S}$.
The matroid $\cM$ has bases that contain exactly one element from $X_1$ and
$n-1$ elements from $X_2$. It is easy to see that any base $B$ has value
$f(B) = 1$, while $f(X_1) = n^2/4$.

\end{document}